\documentclass[journal]{IEEEtran}

\IEEEoverridecommandlockouts
\usepackage{graphicx, subcaption}
\usepackage{amsmath,amssymb}
\usepackage{cite, soul}
\usepackage{url}
\usepackage{bm}
\usepackage{amsthm}

\DeclareMathOperator*{\argmax}{argmax}
\usepackage[usenames,dvipsnames]{color}
\usepackage{comment}
\usepackage{tabularx}

\usepackage[linesnumbered,vlined,ruled]{algorithm2e}
\usepackage{algpseudocode}

\newtheorem{proposition}{Proposition}

\newtheorem{remark}{Remark} 

\title{CRC-Assisted Channel Codes for Integrated Passive Sensing and Communications}
\author{Chenghong Bian, Kaitao Meng, Huihui Wu, Yumeng Zhang, Deniz G{\"u}nd{\"u}z
\thanks{C. Bian and D. G{\"u}nd{\"u}z are with the Department of Electrical and Electronic Engineering, Imperial College London (E-mails: \{c.bian22, d.gunduz\}@imperial.ac.uk). K. Meng is with the Department of Electronic and Electrical Engineering, University College London, London, UK (email: kaitao.meng@ucl.ac.uk). H. Wu is with Biren Tech Research, Shanghai, China (email: huihui.wu@ieee.org). Y. Zhang is with the Department of Electrical and Computer Engineering, Hong Kong University of Science and Technology, Hong Kong (E-mail: eeyzhang@ust.hk).
}
}

\begin{document}

\maketitle
\begin{abstract}
We propose a novel coded integrated passive sensing and communication (CIPSAC) system with orthogonal frequency division multiplexing (OFDM), where a multi-antenna base station (BS) passively senses the parameters of the targets and decodes the information bit sequences transmitted by a user.  
The transmitted signal is comprised of pilot and data OFDM symbols where the data symbols adopt cyclic redundancy check (CRC)-assisted channel codes to facilitate both the decoding and sensing procedures. In the proposed scheme, CRC not only enhances the reliability of communication but also provides guidance to the parameter sensing procedure at the BS.
In particular, a novel iterative parameter sensing and channel decoding (IPSCD) algorithm is proposed, where the correctly decoded codewords that pass CRC are utilized for sensing to improve the parameter estimation accuracy, and in return, more accurate parameter estimates lead to a larger number of correctly decoded data symbols. Conventional sensing algorithms rely only on the received pilot signals, while we utilize both the data and pilot signals for sensing. We provide a detailed analysis of the optimal strategy, in which the wrongly decoded data packets are replaced by zero codewords. 
To further improve the performance, we introduce learning-based near-orthogonal superposition (NOS) codes, which exhibit superior error correction capability especially in the short block length regime. NOS codes are trained using a weighted loss function, where a hyper parameter is introduced to balance the sensing and the communication losses.
Simulation results show the effectiveness of the proposed CIPSAC system and the IPSCD algorithm, where both the sensing and decoding performances are significantly improved with a few iterations. We also carry out extensive ablation studies for a comprehensive understanding of the proposed scheme.
\end{abstract}

\begin{IEEEkeywords} 
Integrated sensing and communication (ISAC), CRC-assisted channel codes, Deep learning, Iterative decoding, Sensing and communication trade-offs. \end{IEEEkeywords}

\section{Introduction}
\label{sec:intro}

Integrated sensing and communications (ISAC) has been gaining significant research interest \cite{isac0, isac1, isac2}. {ISAC  can enhance spectrum and energy efficiency, reduce hardware costs, and address the bandwidth congestion problem in upcoming 6G  networks \cite{network_level, isac_6g}.} ISAC systems can be catagorized as active or passive sensing. In ISAC with active sensing, the transmitted information carrying signal (typically from the BS) serves also for sensing by utilizing its echo to estimate parameters of interest \cite{sense_irs, isac_ris, doa_isac, duplex_isac, liu2021dual}. In passive sensing, on the other hand, the BS uses the received signals from other users to perform joint sensing and decoding without transmitting signals of its own \cite{two_stage_ipsac, isac_learn}.

\begin{figure}
\centering
\includegraphics[width=\columnwidth]{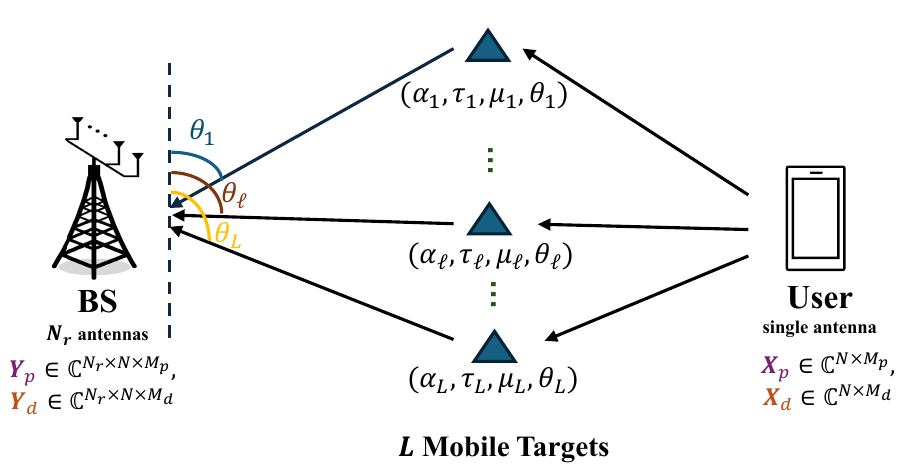}
\caption{In the considered CIPSAC system, the user transmits both pilot and coded data packets to the BS using OFDM. We assume $L$ mobile targets/scatterers, where the $\ell$-th one, $\ell \in [1, L]$, is characterized by its radar cross section (RCS), $\alpha_\ell$, delay, $\tau_\ell$, Doppler, $\mu_\ell$, and AoA, $\theta_\ell$, parameters.}
\label{fig:fig_system}
\end{figure}

{Existing works focus mainly on the theoretical bounds for the mean square error (MSE) of the estimated parameters and the achievable communication rate \cite{isac0, isac1}. However, these communication rates cannot be realized without well-designed channel codes \cite{isac_code, isac_svc, lisac, learn_isac}.} When operating under the finite block length regime, the packet error rate (PER) of the uncoded ISAC system remains at a high level  even in high signal-to-noise ratio (SNR) regimes \cite{crc_isac}. Thus, it is essential to show that the promised gains of ISAC systems holds also with coded signals.  The authors in \cite{isac_code} reveal that satisfactory (active) sensing performance can be achieved if low density parity check (LDPC)-coded signals are utilized for sensing. In \cite{isac_svc}, the authors consider sparse vector codes (SVCs) can be a candidate for active sensing, and show that they outperform convolutional codes in terms of both error correction and sensing performance. In \cite{lisac},  coded waveform design is studied for an active ISAC system with OFDM, where the encoder and decoder are parameterized by recurrent neural networks (RNNs) and are optimized in an end-to-end fashion. All the aforementioned works focus on active sensing, and to the best of our knowledge, no prior work considers coded passive sensing within the ISAC framework. We emphasize that it is not possible to naively apply the code design in \cite{isac_svc, lisac} for active sensing to the considered passive sensing system. Our goal in this paper is to employ coding in ISAC with passive sensing framework. The most relevant works in the current study are \cite{two_stage_ipsac} and \cite{isac_learn}, both of which consider the transmission of uncoded quadrature amplitude modulation (QAM) symbols. {{In particular, in \cite{two_stage_ipsac}, the authors adopt a transformer-based neural network to first detect the QAM symbols and the detected symbols are used for parameter sensing. However, the QAM symbols in \cite{two_stage_ipsac} are uncoded and a considerable amount of symbol errors occur even under a relatively high SNR. These symbol errors will degrade the sensing performance, which further hinders the subsequent communication performance.}} Although the authors of \cite{isac_learn} propose an advanced learning-based iterative data detection and sensing algorithm, the performance is limited due to the lack of error correction.

To the best of our knowledge, we are the first to consider the channel coded integrated passive sensing and communication (CIPSAC) system with OFDM which is a general solution to different setups. In our implementation, illustrated in Fig. \ref{fig:fig_system}, user transmits both pilot and channel coded data packets to a base station over a time-varying multi-path fading channel. The goal of BS is both to decode data packets, and to estimate certain channel parameters, which in turn provide information about targets in the environment. 
The proposed IPSCD algorithm improves the parameter estimation performance despite insufficient number of pilot symbols by treating the correctly decoded data packets as extra pilots.
One challenge that stems from this approach is the potential destructive effect of wrongly decoded packets. This is mitigated by employing CRC in the channel code. In particular, we employ the CRC-assisted sparse regression codes \cite{SPARC2014, HDM2018, nos} with a $K$-best decoding algorithm \cite{mimo_nos} and Polar codes with successive list (SCL) decoding \cite{polar_code, scl_polar} due to their superior performances for short and long data packet transmissions, respectively.

The main contributions of this paper are summarized as follows:

\begin{itemize}
    \item We present a novel IPSCD algorithm for CIPSAC with OFDM, where a user transmits both pilot and data symbols to the BS over a channel induced by multiple targets in the environment. The proposed algorithm is iterative, and the sensing performance can be improved at each iteration by treating the correctly decoded data packets in the previous iteration as pilots, and in return, the more accurate parameter estimates will allow decoding new data packets that could not decoded before. Significant improvements in both the sensing and communication performances  are shown to be achieved within a few iterations.

    \item We provide a theoretical analysis of the parameter sensing procedure, which serves as the core component of the IPSCD algorithm. 
    Compared to the conventional sensing algorithm that relies only on pilot symbols, the core innovation of the proposed sensing algorithm is to replace the wrongly decoded codewords (indicated by CRC flags) using zero codewords for the subsequent processing. We further proved that the proposed solution with zero codewords is optimal for the overall system performance. 

    \item To further improve the performance, we introduce learning-based NOS codes \cite{nos} with a $K$-best decoding algorithm. The flexibility of the CIPSAC system is enhanced by training different NOS codes using loss functions parameterized  by different hyper parameters to achieve different sensing and communication trade-offs.

    \item Rigorous numerical experiments are carried out under different setups and different channel codes. Experimental results demonstrate the effectiveness of the proposed CIPSAC framework over the standard ISAC system. Ablation studies are also carried out  to provide a comprehensive understanding of the proposed CIPSAC system.
    
\end{itemize}

{\it Notations:} Throughout the paper, normal-face letters (e.g., $x$) represent scalars, while uppercase letters (e.g., $X$) represent random variables. Matrices and vectors are denoted by bold {upper} and {lower} case letters (e.g., $\bm{X}$ and $\bm{x}$), respectively.   
$\Re(x)$ ($\Im(x)$) denotes the real (imaginary) part of {a complex variable} $x$.
Transpose and Hermitian operators are denoted by $(\cdot)^\top$, $(\cdot)^\dagger$, respectively. $\text{vec}(\cdot)$ represents the vectorization operation and $\text{diag}(\cdot)$ outputs the diagonal elements of an matrix. {Finally, $\|\bm{S}\|_F$  denotes the Frobneous norm of the matrix $\bm{S}$.}

\section{System Model}\label{sec:System_Model}
In this section, we introduce the standard approach for target sensing and data transmission for the passive ISAC system. The evaluation metric of both the sensing and communication performances are also provided.

\subsection{Pilot Transmission}

We consider the passive ISAC scenario depicted in Fig. \ref{fig:fig_system}, where a user with a single antenna communicates with the BS equipped with $N_r$ antennas via an $L$-tap multi-path fading channel. The $L$ targets can either be static or mobile, and each of them is characterized by its radar cross section (RCS), $\alpha_\ell$, delay, $\tau_\ell$, Doppler, $\mu_\ell$ and angle of arrival (AoA), $\theta_\ell$.

The user adopts OFDM transmission with $N$  OFDM subcarriers and $N_G$ guard subcarriers. We denote the signal bandwidth by $B$, with a delay resolution of $\Delta \tau = \frac{1}{B}$, and the time duration of an OFDM symbol is expressed as $T_{\text{OFDM}} = \frac{N+N_G}{B}$. Note that we assume that the maximal delay of the multi-path channel is shorter than the cyclic prefix (CP) duration.

OFDM signals from the user consist of the pilot and data signals, as detailed below \cite{isac_waveform_design}.
We denote the pilot signal by $\bm{X}_p \in \mathbb{C}^{N\times M_p}$, which consists of $M_p$ consecutive  OFDM symbols. We assume a uniform linear array (ULA) is adopted at the receiver with half-wavelength spacing, i.e., $d = \frac{\lambda}{2}$, where $\lambda$ denotes the wavelength of the transmitted signal.
The Doppler resolution achieved by the pilot signal, $\bm{X}_p$, can be determined as $\Delta f_D =  \frac{1}{M_p T_{\text{OFDM}}}$.  As a result, $N_G$ different delay parameters, $\tilde{\tau}_\ell$, and $M_p$ different Doppler parameters, $\tilde{\mu}_\ell$, can be estimated via $\bm{X}_p$ where we assume $M_p$ is even:
\begin{align}
    \tilde{\tau}_\ell = n_\ell \Delta \tau, \quad &n_\ell \in \{0, \ldots, N_G - 1\}, \notag \\
    \tilde{\mu}_\ell = m_\ell \Delta f_D, \quad &m_\ell \in \{-M_p/2, \ldots, -M_p/2 - 1\}.
    \label{equ:est_dd}
\end{align}
The true delay and Doppler parameters, denoted as $\tau_\ell$ and $\mu_\ell$, lie on the delay-Doppler grid with units $(\Delta \tau', \Delta f_D')$, respectively. Without loss of generality, we assume that the unit of the true delay parameter is identical to $\Delta \tau$ yet the unit of the true Doppler parameter may deviate from $\Delta f_D$:
$$\Big\{\Delta \tau', \Delta f_D'\Big\} = \Big\{\Delta \tau, \frac{1}{M T_{\text{OFDM}}}\Big\},$$
where $M$ stands for the minimum number of OFDM symbols to achieve the required Doppler resolution, $\Delta f_D'$.\footnote{The assumption is well suited for the proposed CIPSAC system in Section \ref{sec:CIPSAC}, in which more correctly decoded OFDM data symbols serve as pilot signals leading to a higher Doppler resolution, $\Delta f_{D}$.} Given the resolution and the maximum detectable delay and Doppler ranges of the pilot OFDM signal, $\bm{X}_p$, indicated by \eqref{equ:est_dd}, the support sets for the true delay and Doppler parameters can be expressed as:\footnote{If the real $\tau, \mu$ parameters exceed the maximum values in \eqref{equ:real_dd}, then, we need to adjust the bandwidth, $B$ and the duration of the OFDM symbol, $T_{\text{OFDM}}$ to enable satisfactory sensing performance.}
\begin{align}
    \tau_\ell = n_\ell \Delta \tau, \quad &n_\ell \in \{0, \ldots, N_G - 1\}, \notag \\
    \mu_\ell = m_\ell \Delta f_D', \quad &m_\ell \in \{-M/2, \ldots, M/2-1\}.
    \label{equ:real_dd}
\end{align}

To achieve accurate delay and Doppler estimation, it is required that the number of OFDM symbols, $M_p$, should be larger than $M$. 
Notice that the number of pilot symbols, $M_p$, is adjustable, and we set $M_p = r M$ with $r\in \mathbb{Z}^+$ to achieve accurate Doppler estimation. It will be shown in Section \ref{sec:CIPSAC} that the condition $M_p = r M$ can be relaxed thanks to additional correctly decoded data OFDM symbols.

The received pilot signal at the receiver in the frequency domain, denoted by $\bm{Y}_p \in \mathbb{C}^{N_r\times N \times M_p}$, can be expressed as:
\begin{align}
    {Y}_p[r,n,m] = & \underbrace{\sum_{\ell=1}^L  \alpha_\ell {X}_p[n,m] \underbrace{e^{-j 2\pi \frac{n_\ell n}{N}}}_{\text{delay}} \underbrace{e^{j 2\pi \frac{m_\ell m}{M}}}_{\text{Doppler}} \underbrace{e^{j \pi r\cos(\theta_\ell)}}_{\text{AoA}}}_{{S}_p[r, n, m]} \notag \\ 
    &+ {W}_p[r,n,m],
    \label{equ:simo_channel}
\end{align}
where $r \in [0, N_r-1], n \in [0, N-1], m \in [0, M_p-1]$, $n_\ell, m_\ell$ are the delay and Doppler parameters introduced in \eqref{equ:real_dd}, $\bm{S}_p \in \mathbb{C}^{N_r \times N \times M_p}$ is the  noiseless received signal and $W_p[r, n, m] \sim \mathcal{CN}(0, \sigma_p^2)$ denotes the independent and identically distributed (i.i.d.) additive white Gaussian noise (AWGN) term. 
The SNR of the pilot signal is defined as
\begin{equation}
    \mathrm{SNR}_p \triangleq \frac{\mathbb{E}(\|\bm{S}_p\|_F^2)}{\mathbb{E}(\|\bm{W}_p\|_F^2)} = 1/\sigma_p^2.
    \label{equ:def_snr}
\end{equation}
Note that the second equation holds as we assume $\mathbb{E}(\|\bm{X}_{p,m}\|_2^2) = N$  and $\alpha_\ell \sim \mathcal{CN}(0, \frac{1}{L}), \ell \in [1, L]$, where $\bm{X}_{p, m}$ is defined as $\bm{X}_{p, m} \triangleq \bm{X}_{p}[:, m]$.

The parameters, $\{\alpha_\ell, \tau_\ell, \mu_\ell, \theta_\ell\}, \ell \in [1, L]$ of the targets can be estimated from the received pilot, $\bm{Y}_p$, via the sensing algorithm presented in Section \ref{sec:CIPSAC}. Let the resultant estimates be denoted by $\{\hat{\alpha}_\ell, \hat{\tau}_\ell, \hat{\mu}_\ell, \hat{\theta}_\ell\}, \ell \in [1, L]$. We also note that, given fixed $\Delta \tau'$ and $\Delta f'_D$, $\hat{\tau}_\ell, \hat{\mu}_\ell$ can be determined by $\hat{n}_\ell, \hat{m}_\ell$, which will be used interchangeably in the following discussions.

\subsection{Data Transmission}\label{sec:data_trans}
First, we present the transmission and decoding of the data packets.
We consider $M_d$ OFDM symbols, each containing information bit sequence, $\bm{b}_b \in \{0, 1\}^{N_b}, b \in [0, M_d-1]$, which is coded and modulated to generate $\bm{x}_{b} \in \mathbb{C}^N$ with a power constraint, $\|\bm{x}_{b}\|^2 \le N$. The modulated codeword is transmitted over the channel\footnote{{We assume the parameters for each of the $L$ targets keep constant for the $M_p$ OFDM symbols for the pilot signal and the subsequent $M_d$ OFDM symbols for the data signal.}} with the same parameters, $\{\alpha_\ell, \tau_\ell, \mu_\ell, \theta_\ell\}, \ell \in [1, L]$ to generate the received signal, $\bm{Y}_d \in \mathbb{C}^{N_r\times N \times M_d}$:
\begin{align}
    {Y}_d[r,n,m] = & \sum_{\ell=1}^L  \alpha_\ell {X}_d[n,m] e^{-j 2\pi \frac{n_\ell n}{N}} e^{j 2\pi \frac{m_\ell m}{M}} e^{j \pi r\cos(\theta_\ell)} \notag \\ 
    &+ {W}_d[r,n,m],
    \label{equ:real_yd}
\end{align}
where $\bm{X}_d$ is obtained by concatenating all the OFDM data symbols, $\bm{X}_d \triangleq [\bm{x}_1, \ldots, \bm{x}_{M_d}]$, while $W_d[r, n, m] \sim \mathcal{CN}(0, \sigma_d^2)$ denotes the i.i.d. AWGN for data transmission and the SNR of the data signal follows the same definition of the pilot signal in \eqref{equ:def_snr}.
After obtaining the estimates of the sensing parameters, $\{\hat{\alpha}_\ell, \hat{\tau}_\ell, \hat{\mu}_\ell, \hat{\theta}_\ell\}$, via the pilot signal, we generate the channel estimate $\hat{\bm{H}}_d \in \mathbb{C}^{N_r \times  N \times M_d}$ as follows:
\begin{align}
    \hat{{H}}_d[r,n,m] & =  \sum_{\ell=1}^L  \hat{\alpha}_\ell  e^{-j 2\pi \frac{\hat{n}_\ell n}{N}}e^{j 2\pi \frac{\hat{m}_\ell m}{M}} e^{j \pi r\cos(\hat{\theta}_\ell)},
    \label{equ:def_csi}
\end{align}
where $\quad r \in [0, N_r-1], n \in [0, N-1], m \in [0, M_d-1]$.
We then decode each one of the $M_d$ information bit sequences $\bm{\hat{b}}_b$, by feeding ${\bm{Y}}_{d, b} \in \mathbb{C}^{N_r \times N}$ and $\hat{\bm{H}}_{d, b} \in \mathbb{C}^{N_r \times N}$ defined as ${\bm{Y}}_{d, b}, \hat{\bm{H}}_{d, b} \triangleq {\bm{Y}}_{d}[:, :, b], \hat{\bm{H}}_{d}[:, :, b]$ to the decoding algorithm detailed in Section \ref{sec:codes_cisac}. 

\subsection{Performance Metrics} \label{sec:sense_para}
After applying the parameter sensing and channel decoding algorithms detailed in Section \ref{sec:CIPSAC} and \ref{sec:codes_cisac}, respectively, to the received pilot and data signals, we obtain the estimated sensing parameters, $\{\hat{\alpha}_\ell, \hat{\tau}_\ell, \hat{\mu}_\ell, \hat{\theta}_\ell\}$, and the decoded bit sequences, $\hat{\bm{b}}_b, b \in [0, M_d-1]$.
For the communication performance, we adopt the  packet error rate (PER) defined as:
\begin{equation}
\text{PER} = \mathbb{E} \left( \frac{1}{M_d} \sum_{b=0}^{M_d-1}  \bm{1}(\bm{{b}}_b \neq \bm{\hat{b}}_b) \right),
\label{equ:per_define}
\end{equation}
where $\bm{1}(\cdot)$ denotes the indicator function, which is 1 if the two bit sequences are identical, and 0 otherwise.

Since the range, velocity and direction of a target can be uniquely determined by the delay, Doppler and the AoA parameters, we use the mean square error (MSE) of these parameters to evaluate the overall sensing performance. Since they are defined similarly, we take the MSE of the delay as an example:
\begin{align}
    \text{MSE}_{\tau} &= \mathbb{E} \left( \sum_{\ell = 1}^L \|n_\ell - \hat{n}_\ell\|^2/N_G^2 \right),
    \label{equ:individual_mse}
\end{align}
where $\tau_\ell = n_\ell \Delta \tau$.
Note that all the three parameters are normalized within the range of [0, 1) to calculate their MSE values. 
We also define the overall MSE as the summation of these three  MSE values:
\begin{equation}
    \text{MSE} \triangleq \text{MSE}_{\tau} + \text{MSE}_{\mu} + \text{MSE}_{\theta},
    \label{equ:overall_mse}
\end{equation}
where $\text{MSE}_{\mu}$ and $\text{MSE}_{\theta}$ denote the MSE of the Doppler and AoA parameters, respectively.
It is worth noting that we can directly sum up the MSE values of different parameters in \eqref{equ:overall_mse} thanks to the normalization in \eqref{equ:individual_mse}.

\begin{figure*}[t]
\centering
\includegraphics[width=\linewidth]{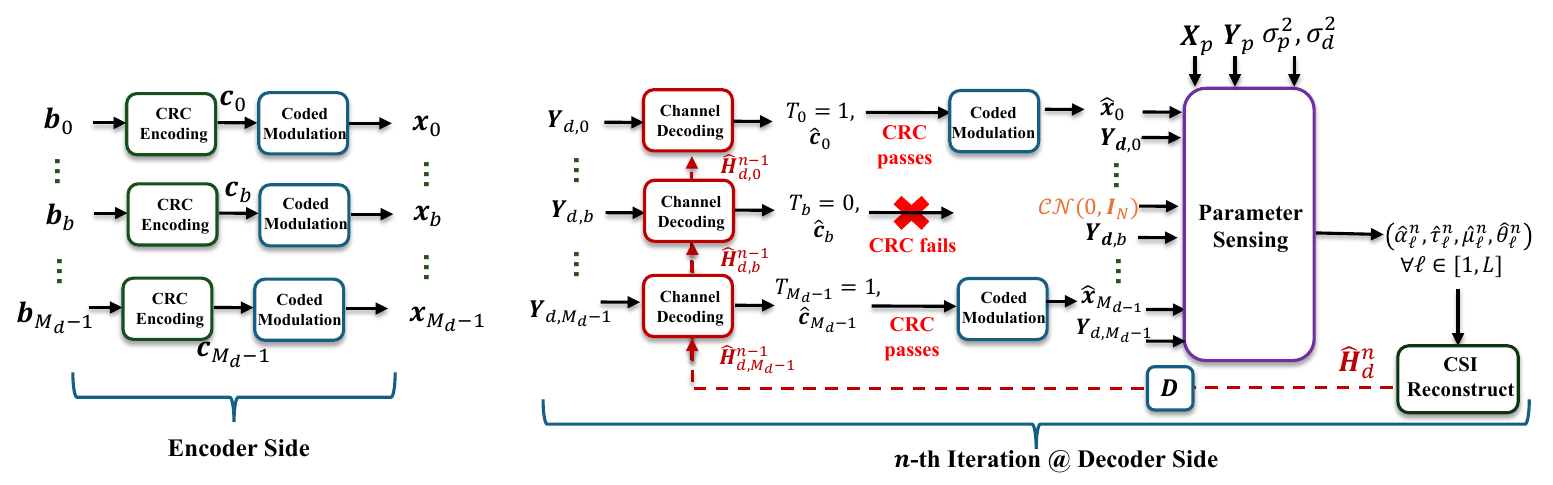}\\
\caption{Illustration of the proposed CIPSAC system. At the encoder, $M_d$ bit sequences are first CRC encoded followed by channel encoding and modulation to generate the data symbols. At the decoder, we run the proposed IPSCD algorithm where for the $n$-th iteration, the channel decoder produces both the decoded bit sequence, $\hat{\bm{c}}_b$ and the CRC flag, $T_b$.  The decoded bit sequence $\bm{\hat{c}}_b$ with $T_b = 1$ are re-encoded yet the failed ones with $T_b = 0$ are replaced by randomly generated signals for the next round of parameter sensing process.}
\label{fig:ICED}
\end{figure*}

\section{The Proposed CIPSAC Framework}\label{sec:CIPSAC}
In this section, we illustrate the proposed CIPSAC framework where we implement the channel coding and modulation at the transmitter and iterative parameter sensing and channel decoding (IPSCD) algorithm at the receiver. We consider a scenario where a single antenna user communicates with a multi-antenna BS over the corresponding time-varying single input multiple output (SIMO-OFDM) channel characterized by mobile targets/scatterers.

\subsection{CRC-assisted Channel Encoding and Modulation}
The overall encoding process is comprised of two parts, namely, the CRC encoding and the subsequent channel coded modulation. To start with, the input bit sequence, $\bm{b} \in \{0, 1\}^{N_b}$, is fed to the CRC encoding block, where $N_{crc}$ extra CRC bits are appended to $\bm{b}$ to produce a CRC-encoded bit sequence, $\bm{c} \in \{0, 1\}^{N_b+N_{crc}}$, for the subsequent channel encoding and modulation process, denoted by $f(\cdot)$. The transmitted codeword, $\bm{x} \in \mathbb{C}^N$, can be expressed as:
\begin{equation}
\bm{x} = f(\bm{c}),
\label{equ:channel_enc}
\end{equation}
where $\bm{x}$ satisfies the power constraint in Section \ref{sec:data_trans}. We note that different $f(\cdot)$ realizations correspond to different coded modulation schemes. In this paper, we will consider Polar codes with QPSK modulation, sparse regression  code (SPARC) \cite{SPARC2014} and learning-based NOS codes \cite{nos} as detailed in Section \ref{sec:codes_cisac}.

\subsection{Iterative Parameter Sensing and Channel Decoding (IPSCD) Algorithm}
Here, we present the proposed IPSCD algorithm adopted at the receiver, where both the communication and sensing performances can be significantly improved with a few iterations. In the standard ISAC algorithm presented in Section \ref{sec:System_Model}, the BS estimates the parameters, $\{\hat{\alpha}_\ell, \hat{\tau}_\ell, \hat{\mu}_\ell, \hat{\theta}_\ell\}$, and reconstructs the estimated channel matrix, $\hat{\bm{H}}_d$, for the OFDM data symbols using  the $M_p$   OFDM pilot symbols, and the subsequent $M_d$  data symbols are decoded using $\hat{\bm{H}}_d$. However, this will result in a poor performance especially when $M_p$ is less than the minimum number of OFDM symbols to achieve the actual Doppler resolution, $\Delta f'_D$, i.e., $M_p < M$. In this case, each of the Doppler estimates, $\hat{\mu}_\ell$, deviates significantly from the ground truth, $\mu_\ell$, leading to less satisfactory sensing and decoding performances. To mitigate this, in the IPSCD algorithm, we treat the successfully decoded data packets as pilot signals, and use them to generate more accurate parameter estimates, which will help improve the decoding performance. The overall flowchart of the proposed scheme is shown in Fig. \ref{fig:ICED} and is detailed as follows.

As shown in the left hand side of Fig. \ref{fig:ICED}, $M_d$  transmitted symbols, $\bm{x}_b, b\in [0, M_d-1]$, are generated according to \eqref{equ:channel_enc}.  After passing through the SIMO-OFDM channel, the receiver has access to $\bm{Y}_{d, b} \in \mathbb{C}^{N_r\times N}, b \in [0, M_d-1]$, as well as the received pilot signal, $\bm{Y}_p \in \mathbb{C}^{N_r\times N \times M_p}$. 
The IPSCD algorithm starts by initializing the CRC flag and the decoded bit sequence for each of the $M_d$ data packets to $\{T_b = 0, \hat{\bm{c}}_b = \bm{0}\}_{b \in [0, M_d-1]}$, respectively, where $T_b \in \{0, 1\}$ is the CRC flag. 
The BS  estimates the sensing parameters, $\{\hat{\alpha}^0_\ell, \hat{\tau}^0_\ell, \hat{\mu}^0_\ell, \hat{\theta}^0_\ell\}, \ell \in [1, L]$ as follows:
\begin{equation}
    \{\hat{\alpha}^0_\ell, \hat{\tau}^0_\ell, \hat{\mu}^0_\ell, \hat{\theta}^0_\ell\}_{\ell \in [1, L]} = h(\bm{Y}_p, \bm{Y}_d, \bm{X}_p, \hat{\bm{X}}_d^{0}, \bm{T}),
    \label{equ:para_sensing}
\end{equation}
where $h(\cdot)$ represents the sensing procedure detailed later in Section \ref{sec:sense_alg}, $\hat{\bm{X}}_d^{0}$ is the estimate of the transmitted data symbols at the receiver, which can be initialized using an arbitrary matrix since none of the $M_d$ data packets has been successfully decoded and the superscript of the parameters denotes the iteration index. To facilitate the following illustrations, we assume the entries of $\hat{\bm{X}}_d^{0}$ are i.i.d., each following a complex Gaussian distribution, $\mathcal{CN}(0, 1)$.
With the estimated sensing parameters in \eqref{equ:para_sensing}, we reconstruct the channel estimate for the data packets,  denoted by $\bm{\hat{H}}_d^{0}$, according to \eqref{equ:def_csi}:
\begin{align}
    \hat{{H}}^0_d[r,n,m] & =  \sum_{\ell=1}^L  \hat{\alpha}^0_\ell  e^{-j 2\pi \frac{\hat{n}^0_\ell n}{N}}e^{j 2\pi \frac{\hat{m}^0_\ell m}{M}} e^{j \pi r\cos(\hat{\theta}^0_\ell)},
    \label{equ:csi_rec}
\end{align}
where $\hat{n}^0_\ell, \hat{m}^0_\ell = \hat{\tau}^0_\ell/\Delta \tau, \hat{\mu}^0_\ell/\Delta f_D'$ as in \eqref{equ:real_dd}.

In the first iteration, the data symbols are decoded using $\bm{\hat{H}}^{0}_d$ in \eqref{equ:csi_rec}. The decoding process for the $b$-th data packet can be expressed as:
\begin{equation}
\{T_b, \hat{\bm{c}}_b\} = g(\bm{Y}_{d, b},  {\bm{\hat{H}}^{0}_{d, b}}),
\label{equ:channel_dec}
\end{equation}
where $g(\cdot)$ denotes the  CRC-assisted channel decoding algorithm corresponding to the encoding function, $f(\cdot)$ in \eqref{equ:channel_enc}, $T_b$ is the CRC flag and $\bm{Y}_{d, b},  {\bm{\hat{H}}^{0}_{d, b}} \triangleq \bm{Y}_{d}[:,:,b],  {\bm{\hat{H}}^{0}_{d}[:,:,b]}$. It is worth mentioning that having $T_b = 1$ does not guarantee successful decoding of the packet, i.e., there exists a scenario where $T_b = 1$ yet $\hat{\bm{c}}_b \neq \bm{c}_b$. This is because some wrongly decoded data packets may also pass the CRC. We term this event as an `outage' and the outage probability is denoted by $P_o$. Since $P_o$ is typically small, it is plausible to assume that the packet with $T_b = 1$ is correctly decoded and the output bit sequence, $\hat{\bm{c}}_b$, will be used to estimate the parameters for the next iteration.\footnote{It is shown in the experiment that the wrongly decoded packets that pass the CRC have limited effect on the final performance.}

{After decoding, the decoded bit sequence, $\hat{\bm{c}}_b$, with $T_b = 1$ is re-encoded into the transmitted signal, defined as $\hat{\bm{x}}_{b}$. Then, the estimate of the transmitted data packets for the first iteration, $\hat{\bm{X}}_d^{1} \in \mathbb{C}^{N\times M_d}$, can be represented as follows:}
\begin{equation}
    \hat{\bm{X}}_{d, b}^{1} = \left\{
    \begin{aligned}
    \hat{\bm{x}}_{b} & , & T_b = 1, \\
    \mathcal{CN}(\bm{0}, \bm{I}_N) & , & T_b = 0.
    \end{aligned}
    \right.
    \label{equ:def_hat_X}
\end{equation}

After obtaining $\hat{\bm{X}}_d^{1}$, we update the sensing parameters for the first iteration as:
\begin{equation}
    \{\hat{\alpha}^1_\ell, \hat{\tau}^1_\ell, \hat{\mu}^1_\ell, \hat{\theta}^1_\ell\}_{\ell \in [1, L]} = h(\bm{Y}_p, \bm{Y}_d, \bm{X}_p, \hat{\bm{X}}_d^{1}, \bm{T}),
    \label{equ:para_sensing_1}
\end{equation}
which are then adopted to generate the channel estimate $\bm{\hat{H}}_d^{1}$ according to \eqref{equ:csi_rec} for channel decoding of the second iteration. We further note that, for the $i$-th iteration with $i>1$,  we only need to decode the data packets with $T_b = 0$ for reduced decoding complexity. The IPSCD algorithm terminates if the number of iterations reaches $N_{i}$, or all the $M_d$ data packets pass CRC. We summarized the proposed IPSCD algorithm in Algorithm \ref{algorithm:idce}.

\begin{algorithm}
	\caption{The proposed IPSCD algorithm.}
    \label{algorithm:idce}

	\SetKwInOut{Input}{Input}\SetKwInOut{Output}{Output}
	\SetKwFunction{UpdateCodebook}{UpdateCodebook}
	\SetKwFunction{SPARCEncode}{SPARCEncode}
	\SetKwFunction{Construct}{Construct}
	
	\Input{$\bm{X}_p, \bm{Y}_p, \bm{Y}_d, K, {N_{i}}, M_d$}
	\Output{$\{\hat{\bm{c}}_b\}_{b \in [0, M_d-1]}, \{\hat{\alpha}^{N_{i}}_\ell, \hat{\tau}^{N_{i}}_\ell, \hat{\mu}^{N_{i}}_\ell, \hat{\theta}^{N_{i}}_\ell\}_{\ell \in [1, L]}, $} 
	
	\BlankLine
        Initialize the CRC flags and the estimate of the transmitted signals, $\bm{T} \leftarrow  \bm{0}_{M_d}, \hat{{X}}_d^{0}[n, b] \sim \mathcal{CN}(0, 1)$,\\
        Estimate $\{\hat{\alpha}^0_\ell, \hat{\tau}^0_\ell, \hat{\mu}^0_\ell, \hat{\theta}^0_\ell\}_{\ell \in [1, L]} = h(\bm{Y}_p, \bm{Y}_d, \bm{X}_p, \hat{\bm{X}}_d^{0}, \bm{T})$,\\
        Reconstruct the channel estimate, $\hat{{H}}^0_d[r,n,m]  =  \sum_{\ell=1}^L  \hat{\alpha}^0_\ell  e^{-j 2\pi \frac{\hat{n}^0_\ell n}{N}}e^{j 2\pi \frac{\hat{m}^0_\ell m}{M}} e^{j \pi r\cos(\hat{\theta}^0_\ell)}$,\\
	
        \For{$i=1$ \KwTo {$N_i$}}{
            Initialize $\hat{\bm{X}}_d^{i} \leftarrow \bm{0}_{N\times M_d}$, \\
		\For{$b=0$ \KwTo $M_d-1$}{
             \eIf{$T_b = 0$}{
                $\{T_b, \hat{\bm{c}}_b\} = g(\bm{Y}_{d, b},  {\bm{\hat{H}}^{i-1}_{d, b}})$, \\
                \eIf{$T_b=0$}{
                $\hat{\bm{X}}_{d, b}^{i} \leftarrow \mathcal{CN}(\bm{0}, \bm{I}_N)$,\\
                }
                {$\hat{\bm{X}}_{d, b}^{i} \leftarrow f( \hat{\bm{c}}_b)$,\\}}
            {
            $\hat{\bm{X}}_{d, b}^{i} \leftarrow \hat{\bm{X}}_{d, b}^{i-1}$,\\}
		}

        Estimate $\{\hat{\alpha}^i_\ell, \hat{\tau}^i_\ell, \hat{\mu}^i_\ell, \hat{\theta}^i_\ell\}_{\ell \in [1, L]} = h(\bm{Y}_p, \bm{Y}_d, \bm{X}_p, \hat{\bm{X}}_d^{i}, \bm{T})$,\\
        Reconstruct $\hat{{H}}^i_d[r,n,m]  =  \sum_{\ell=1}^L  \hat{\alpha}^i_\ell  e^{-j 2\pi \frac{\hat{n}^i_\ell n}{N}}e^{j 2\pi \frac{\hat{m}^i_\ell m}{M}} e^{j \pi r\cos(\hat{\theta}^i_\ell)}$.

	}

	\BlankLine
\end{algorithm}

\subsection{Sensing Algorithm}\label{sec:sense_alg}
Next, we present the sensing function, $h(\cdot)$, which outputs sensing parameters using the original pilot signals as well as the correctly decoded data packets. The sensing algorithm for the static scenario is presented in \cite{crc_isac} and can be understood as a special case of the considered scenario with $L$ mobile targets and the receiver is equipped with $N_r>1$ antennas. As shown in Fig. \ref{fig:sensing_flowchart}, the sensing function, $h(\cdot)$, takes $\bm{Y}_p \in \mathbb{C}^{N_r\times N\times M_p}, \bm{Y}_d \in \mathbb{C}^{N_r\times N\times M_d}, {\bm{X}}_p \in \mathbb{C}^{N\times M_p}, \hat{\bm{X}}_d \in \mathbb{C}^{N\times M_d}$ and $\bm{T} \in \{0, 1\}^{M_d}$ as input, and outputs the estimates of the sensing parameters, $\{\hat{\alpha}_\ell, \hat{\tau}_\ell, \hat{\mu}_\ell, \hat{\theta}_\ell\}_{\ell \in [1, L]}$. We illustrate the estimation process for each sensing parameter as follows.  

\begin{figure*}[t]
\centering
\includegraphics[width=0.9\linewidth]{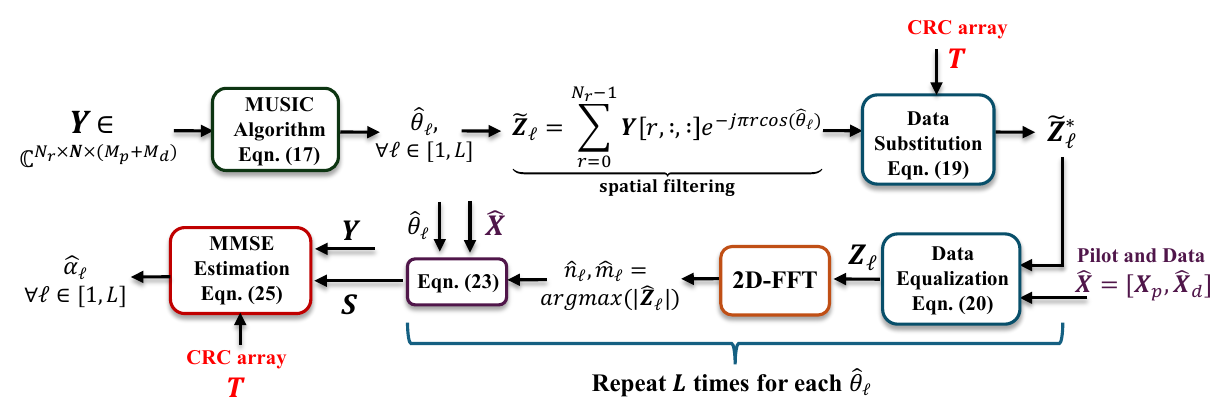}\\
\caption{Illustration of the sensing procedure of the proposed IPSCD algorithm. In particular, we first apply MUSIC algorithm to the input signal ${\bm{Y}}$ to obtain the AoA estimates, $\hat{\theta}_\ell, \ell \in [1, L]$. Then, for the $\ell$-th AoA estimate, $\hat{\theta}_\ell$, we perform spatial filtering, data substitution, and data equalization detailed in Section \ref{sec:sense_alg} to obtain $\bm{Z}_\ell \in \mathbb{C}^{N \times (M_p + M_d)}$. After applying the 2D-FFT to $\bm{Z}_\ell$, the delay and Doppler parameters, $\hat{\tau}_\ell, \hat{\mu}_\ell$ (or equivalently, $\hat{n}_\ell, \hat{m}_\ell$), associated with the AoA estimate, $\hat{\theta}_\ell$ can be obtained. Finally, the RCS parameters are calculated according to Eqn. \eqref{equ:est_alpha}.}
\label{fig:sensing_flowchart}
\end{figure*}

\subsubsection{AoA estimate}
To start with, the standard multiple signal classification (MUSIC) algorithm is applied to the received signals, $\bm{Y} \in \mathbb{C}^{N_r  \times N \times (M_p+M_d)}$, which is obtained by concatenating the received pilot, $\bm{Y}_p$, and the data signals, $\bm{Y}_d$, along the last dimension regardless of whether the data packets within $\bm{Y}_d$ are successfully decoded or not. We first generate the covariance matrix:
\begin{equation}
    \bm{R} \triangleq \frac{1}{(M_p+M_d)N} \bm{{Y}}_r\bm{{Y}}_r^\dagger,
\end{equation}
where $\bm{{Y}}_r \in \mathbb{C}^{N_r\times (M_p+M_d)N}$ is obtained by reshaping $\bm{Y}$. Then, the eigenvalue decomposition is applied to $\bm{R}$, and the noise subspace, denoted by $\bm{E} \in \mathbb{C}^{(N_r-L)\times N_r}$, is obtained by taking the eigenvectors corresponding to the $N_r-L$ smallest eigenvalues. Denote the array response vector by $\bm{a}(\psi) = [1, e^{j\pi \cos(\psi)}, \ldots, e^{j(N_r-1)\pi \cos(\psi)}]^\top$, the spectrum of the MUSIC algorithm can be expressed as:
\begin{equation}
    \text{Spec}(\psi) = \frac{1}{\bm{a}^\dagger(\psi) \bm{E}^\dagger \bm{E} \bm{a}(\psi)},
    \label{equ:music_alg}
\end{equation}
and the angles corresponding to the largest $L$ peaks of $\text{Spec}(\psi)$ serve as the AoA estimates, $\hat{\theta}_\ell, \ell \in [1, L]$, of the target.

\subsubsection{Delay and Doppler estimates}
We then try to figure out the delay and Doppler parameters for each of the $\hat{\theta}_\ell$ values.
To achieve this, we first apply spatial filtering where the matched filtered output between ${\bm{Y}} \in \mathbb{C}^{N_r \times N \times (M_p + M_d)}$ and the steering vector, $\bm{a}(\hat{\theta}_\ell)$, is computed to obtain the two-dimensional (2D) matrix, $\widetilde{\bm{Z}}_\ell \in \mathbb{C}^{N\times (M_p+M_d)}$:
\begin{equation}
    \widetilde{\bm{Z}}_\ell = \sum_{r = 0}^{N_r-1} {\bm{Y}}[r, :, :]e^{-j\pi r \cos(\hat{\theta}_\ell)},
    \label{eqqu:def_tilde_Z}
\end{equation}
which will be used for subsequent processing.

\textbf{Data substitution}:
However, unlike \eqref{equ:music_alg} where we directly feed $\bm{Y}$ to the MUSIC algorithm for AoA estimation, extra processing for the matrix $\widetilde{\bm{Z}}_\ell$ is required due to the fact that the estimation accuracy of the delay and Doppler parameters strongly depend on the correct decoding of the received data symbols, $\bm{Y}_d$. When $T_b = 0$, i.e., a decoding error happens,   the $b$-th   OFDM data symbol is no longer capable for sensing. In the static scenario illustrated in \cite{crc_isac}, the packet with $T_b = 0$ is directly discarded and the cardinality {of the latent $\widetilde{\bm{Z}}_\ell$ is reduced}. In the mobile case, however, we estimate the delay and Doppler parameters by applying efficient 2D-FFT algorithm\footnote{We can also adopt the maximum likelihood (ML) algorithm, however, its complexity is prohibitive in practice.} to the latent, $\widetilde{\bm{Z}}_\ell$, thus, it is impossible to directly remove the data packets with $T_b = 0$ as the 2D-FFT algorithm fails if the dimensionality of the latent matrix reduces.  To this end, we introduce $\bm{u}_m \in \mathbb{C}^{N}$ to replace the wrongly decoded symbols in $\widetilde{\bm{Z}}_\ell$ as follows:
\begin{equation}
    \widetilde{\bm{Z}}^*_\ell = \left\{
    \begin{aligned}
    \widetilde{\bm{Z}}_{\ell, m}  & , & m \in [0, M_p-1] \; || \; T_{m-M_p} = 1, \\
    \bm{u}_m & , & m \ge M_p \; \& \; T_{m-M_p} = 0,
    \end{aligned}
    \right.
    \label{equ:def_tilde_Z}
\end{equation}
where $\widetilde{\bm{Z}}_{\ell, m} \triangleq \widetilde{\bm{Z}}_\ell[:, m], m\in [0, M_p+M_d-1]$.

\textbf{Data equalization}:
Before estimating the delay and Doppler parameters, we equalize $\widetilde{\bm{Z}}^*_\ell$ using the pilot signal, $\bm{X}_p$, and the estimated data symbols, $\hat{\bm{X}}_{d, b}$, to generate ${\bm{Z}}_\ell$:
\begin{equation}
    {\bm{Z}}_{\ell, m} = \left\{
    \begin{aligned}
    \widetilde{\bm{Z}}^*_{\ell, m}/\bm{X}_{p, m}  & , & m \in [0, M_p-1], \\
    \widetilde{\bm{Z}}^*_{\ell, m}/\hat{\bm{X}}_{d, m-M_p} & , & m \ge M_p.
    \end{aligned}
    \right.
    \label{equ:def_Z}
\end{equation}
It is worth mentioning that different $\bm{u}_m$ realizations make a difference and there exists an optimal solution which is summarized in the next proposition:

\begin{proposition}
Among all randomly chosen $\bm{u}_m$ realizations, $\bm{u}_m = \bm{0}_{N}$ is the optimal solution to achieve the best peak-to-side-lobe-ratio (PSR) for the delay and Doppler estimation.
\label{pro:pro1}
\end{proposition}

\begin{proof}
We refer to Appendix \ref{sec:APPA} for the proof.
\end{proof}

Intuitively, when a decoding error happens, i.e., $T_{m-M_p} = 0$, the estimate of the transmitted signal, $\hat{\bm{X}}_{d, m-M_p}$, may deviate significantly from the actual one. Thus, during the \textbf{data equalization} procedure, $\forall \bm{u}_m \neq \bm{0}_N$ realization may produce outlier ${\bm{Z}}_{\ell, m}$ values due to the wrong $\hat{\bm{X}}_{d, m-M_p}$ estimate which would significantly hinder the subsequent 2D-FFT procedure. Having $\bm{u}_m = \bm{0}_N$, on the other hand, not only mitigates the performance degradation caused by the wrongly decoded signals but also provides an efficient implementation even without the need to generate random signals.

We then apply 2D-FFT to $\bm{Z}_\ell$, and the estimated delay and Doppler parameters, $\hat{n}_\ell, \hat{m}_\ell$ (which are equivalent to $\hat{\tau}_\ell, \hat{\mu}_\ell$ according to \eqref{equ:real_dd}), are determined by finding the indices of the largest absolute value of the 2D-FFT output, $\hat{\bm{Z}}_\ell$:
\begin{equation}
    \hat{n}_\ell, \hat{m}_\ell = \argmax_{n, m} |\hat{\bm{Z}}_\ell|.
    \label{equ:est_tau_mu}
\end{equation}
where $n \in [0, N_G], m \in [-M/2, M/2-1]$.

\begin{remark}
    We note that the estimation of $\hat{m}_\ell$ in \eqref{equ:est_tau_mu} corresponds to the special case where $M = M_p + M_d$. For a more general setup where $M_p + M_d = rM, r\in \mathbb{Z}^+$, $\hat{m}_\ell$ should be scaled by $r$ for its actual range.
\end{remark}

\begin{remark}
    It is worth mentioning that the sensing algorithm mentioned above to obtain the estimates of the AoA, delay and Doppler parameters can be generalized to the case where only $\bm{Y}_p$ is available to the receiver. In this case, even when $M_p < M$,\footnote{Here, if we only use $M_p$ pilot OFDM symbols to perform sensing, the resolution of the estimated Doppler parameter is below the actual one leading to less satisfactory estimation and decoding performance.} it is indicated in \eqref{equ:def_tilde_Z} that by applying zero-padding, we can also obtain $\bm{Z}_\ell$ with dimension $N \times M$ which achieves the required Doppler resolution, $\Delta f_D'$.
\end{remark}

\subsubsection{RCS estimate}
Combining \eqref{equ:music_alg} and \eqref{equ:est_tau_mu}, we obtain the AoA, delay, and Doppler estimates, $\{\hat{\theta}_\ell, \hat{\tau}_\ell, \hat{\mu}_\ell\}, \ell \in [1, L]$ for all the $L$ targets. To generate the channel estimate, $\hat{\bm{H}}_d$ in Algorithm \ref{algorithm:idce}, it is essential to estimate the RCS, $\alpha_\ell$, for each of the $\ell$-th target. 

To start with, we first notice that we can re-organize the channel input and output relationship described in \eqref{equ:simo_channel}, and express the received signal $\bm{Y}$ as follows:
\begin{align}
    \bm{Y} = [\bm{S}_1, \ldots, \bm{S}_L] \bm{\alpha}+ \bm{W},
    \label{equ:real_y}
\end{align}
where $\bm{\alpha} \triangleq [\alpha_1, \ldots, \alpha_L]^\top$ and $\bm{S}_\ell \in \mathbb{C}^{N_r \times N \times (M_p + M_d)}$ with
\begin{align}
    {S}_\ell[r, n, m] = {X}[n, m]e^{-j 2\pi \frac{n_\ell n}{N}} e^{j 2\pi \frac{m_\ell m}{M}} e^{j \pi r\cos(\theta_\ell)}.
    \label{equ:gen_s}
\end{align}
Note that $\bm{X} \triangleq [\bm{X}_p, \bm{X}_d]$ is obtained by concatenating the pilot and data OFDM symbols.
Similar to \cite{crc_isac}, we generate the MMSE estimate of $\bm{\alpha}$ based on Eqn. \eqref{equ:real_y}. However, since we have no access to the transmitted codeword, $\bm{X}_{d, b}$, with $T_b = 0$, we remove the corresponding entries from each of $\bm{S}_\ell, \ell \in [1, L]$, the received signal, $\bm{Y}$, and the noise tensor, $\bm{W}$, to obtain $\widetilde{\bm{S}}_\ell$, $\widetilde{\bm{Y}}$ and $\widetilde{\bm{W}}$, respectively. All these tensors are of the same dimension, i.e., $N_r \times N \times (M_p+n)$ where $n = \sum_{b=0}^{M_d-1} T_b$. After vectorizing $\widetilde{\bm{S}}_\ell$, $\widetilde{\bm{Y}}$ and $\widetilde{\bm{W}}$, \eqref{equ:real_y} transforms into:
\begin{equation}
    \tilde{\bm{y}} = \underbrace{[\tilde{\bm{s}}_1, \ldots, \tilde{\bm{s}}_L]}_{\widetilde{\bm{S}}} \bm{\alpha} + \tilde{\bm{w}},
\end{equation}
where $\tilde{\bm{y}}, \tilde{\bm{s}}_\ell, \tilde{\bm{w}} = \text{vec}(\widetilde{\bm{Y}}), \text{vec}(\widetilde{\bm{S}}_\ell), \text{vec}(\widetilde{\bm{W}})$ are the vectorized version of the tensors and $\widetilde{\bm{S}} \in \mathbb{C}^{N_rN(M_p+n)\times L}$. Then, $\bm{\alpha}$ can be estimated via standard MMSE algorithm:
\begin{equation}
    \hat{\bm{\alpha}} = \widetilde{\bm{S}}^\dagger (\widetilde{\bm{S}} \widetilde{\bm{S}}^\dagger + \bm{R}_{\widetilde{\bm{W}}})^{-1} \tilde{\bm{y}},
    \label{equ:est_alpha}
\end{equation}
where $\bm{R}_{\widetilde{\bm{W}}}$ denotes the covariance matrix of $\widetilde{\bm{W}}$ satisfying $$\text{diag}(\bm{R}_{\widetilde{\bm{W}}}) = [\underbrace{\sigma_p^2, \ldots, \sigma_p^2}_{\times N_rNM_p}, \quad \underbrace{\sigma_d^2, \ldots, \sigma_d^2}_{\times N_rN n}].$$
After obtaining \eqref{equ:est_alpha}, all the parameters, $\{\hat{\alpha}_\ell, \hat{\tau}_\ell, \hat{\mu}_\ell, \hat{\theta}_\ell\}, \ell \in [1, L]$ are obtained and the corresponding reconstructed channel estimate is utilized for the next round of channel decoding. The overall flowchart of the sensing algorithm is shown in Fig. \ref{fig:sensing_flowchart}.


\newcommand{\y}{\bm{y}}
\newcommand{\A}{\bm{A}}
\newcommand{\x}{\bm{x}}
\newcommand{\Sig}{{S}}
\newcommand{\HM}{\bm{H}}
\newcommand{\tabincell}[2]{\begin{tabular}{@{}#1@{}}#2\end{tabular}} 

\begin{algorithm}
	\caption{CRC-assisted $K$-best decoding algorithm for SIMO-OFDM channel.}
    \label{algorithm:decode}

	\SetKwInOut{Input}{Input}\SetKwInOut{Output}{Output}
	\SetKwFunction{KBest}{KBest}
	\SetKwFunction{CRCcorrection}{CRCcorrection}
	\SetKwFunction{IdxToBits}{IdxToBits}
	\SetKwFunction{CRCDecode}{CRCDecode}
	\SetKwFunction{Reorder}{Reorder}
	\SetKwFunction{ChooseLayer}{ChooseLayer}
	\SetKwFunction{SelectNodes}{SelectNodes}
	\SetKwFunction{SelectDistinctNodes}{SelectDistinctNodes}
 \SetKwFunction{LoopedKBest}{LoopedKBest}
	\SetKwData{list}{outputList}
	\SetKwData{idx}{idx}
	\SetKwData{errFlag}{errFlag}
	\SetKwData{decodedBits}{decodedBits}
	\SetKwData{anc}{anc}
	
	\Input{$K, {N_{i}^d}, \bm{y},  {{\{\bm{\mathcal{C}}_{\bm{H}}\}}}$}
	\Output{$T, \bm{\hat{c}}$} 
	
	\BlankLine
	\For{$k=1$ \KwTo $K$}{
		$\bm{u}(k), s(k), \idx(k) \leftarrow 0, 0, [\;]$ \\ 
		
	}
    $\mathcal{O} \leftarrow [\;]$ \hfill (empty decoded layer index)\\
    $T \leftarrow 0$ \hfill (CRC flag)\\
        \%\% \textbf{{Original $K$-best decoding:}}\\
	\For{${j}=1$ \KwTo {$V$}}{
	    ${l_j} \leftarrow$ \ChooseLayer{$\overline{\mathcal{O}}$}\\
	    $\mathcal{O} \leftarrow$ $[\mathcal{O}, {l_j}]$\\
		\For{$k=1$ \KwTo $K$}{
			
			 {{$\textbf{s}^{tmp}(k) \leftarrow s(k) - 2\Re(\y^\dagger\bm{{\cal C}}_{l_j,\bm{H}}-\bm{u}^\dagger(k)\bm{{\cal C}}_{l_j,\bm{H}}) + diag(\bm{{\cal C}}_{l_j,\bm{H}}^\dagger\bm{{\cal C}}_{l_j,\bm{H}})$}}\\
		}
		$[\bm{s}, \idx_{new}, \anc] \leftarrow$ \SelectNodes{$\bm{s}^{tmp}, K$}\\
		\For{$k=1$ \KwTo $K$}{
			 {{$\bm{u}(k) \leftarrow \bm{u}(\anc(k)) + \bm{{\cal C}}_{l_j,\bm{H}}[\idx_{new}(k)] $}}\\
			$\idx(k) \leftarrow [\idx (\anc(k)), \idx_{new}(k)]$\\
		}
	}
        \%\% \textbf{{Looped $K$-best decoding:}}\\
	$\idx$ = \LoopedKBest($\idx, \{\bm{u}^{(V)}_k, \bm{s}_k^{(V)}\}_{k \in [1, K]}, \mathcal{O}, N_i^d)$ \\
	\list $\leftarrow$ \Reorder{$\idx$, $\mathcal{O}$}
	
        \%\% \textbf{CRC decoding:}\\
	\While{$T \neq 1$ and $k \le K$}{
		$\bm{\hat{c}}$ $\leftarrow$ \IdxToBits{\list$(k)$}\\
		$T$ $\leftarrow$ \CRCDecode{$\bm{\hat{c}}$}\\
	}

	\BlankLine
	\BlankLine
\end{algorithm}

\section{Learning a Channel Code for CIPSAC}\label{sec:codes_cisac}
In this section, we first introduce the training procedure of the learning-based NOS code \cite{nos} for the considered CIPSAC system. In particular, the loss function is parameterized by different hyper parameters to achieve different sensing and communication trade-offs. To further improve the PER performance, the CRC-assisted $K$-best decoding algorithm originally proposed in \cite{nos} is extended to the considered SIMO-OFDM channel.

\begin{figure}
\centering
\includegraphics[width=\columnwidth]{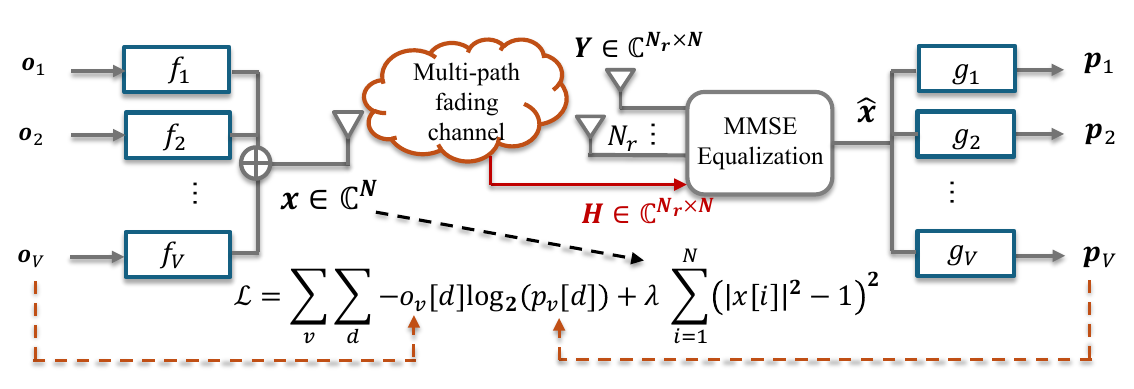}
\caption{The encoder and decoder of the NOS code parameterized by neural networks over the SIMO-OFDM channel. The overall loss function is a combination of the communication and the sensing losses weighted by a hyper parameter, $\lambda$.}
\label{fig:fig_nos}
\end{figure}

\subsection{Learning-based NOS Code with Communication and Sensing Trade-offs}\label{sec:nos_code}

\subsubsection{Learning a NOS code}
Though it has been shown in \cite{SPARC2014} that adopting the random Gaussian codebook achieves satisfactory decoding performances, there is no guarantee that the codebook is optimal, especially in the short block length regime and for the considered scenario. Thus, we try to improve both the sensing and decoding performances by adopting a codebook obtained via end-to-end training. The neural network architecture of the NOS code over an SIMO-OFDM channel is shown in Fig. \ref{fig:fig_nos} and is illustrated as follows.

\textbf{NOS encoding:}
To start with, we transform the $V$ indices, $d_v \in [1, D]$, into one-hot vectors, {$\bm{o}_v$}, which are fed to the corresponding neural network-based encoders {$f_v$} to generate real-valued vectors {$\tilde{\bm{x}}_v={f_v}(\bm{o}_v)$} of length ${2N}$. Each ${f_v, v\in [1,V]}$ has the same neural network architecture \cite{mimo_nos} that consists of linear layers, batch normalization layers, and non-linear activation functions. Since each {$\tilde{\bm{x}}_v$} conveys the same amount of information,  {we assign the same energy {$||\tilde{\bm{x}}_v||^2_2=\frac{{N}}{V}$} to each {$\tilde{\bm{x}}_v$}.} Instead of transmitting a real-valued signal, we convert the length-${2N}$ real-valued vector {$\tilde{\bm{x}}_v$} into a complex vector {$\bm{x}_v = \tilde{\bm{x}}_v^R + j \tilde{\bm{x}}_v^I$} {where $\tilde{\bm{x}}_v^R$ and $\tilde{\bm{x}}_v^I$ are of dimension-$N$ which correspond to the first and the remaining $N$ elements of $\tilde{\bm{x}}_v$, respectively.} The superimposed signal $\bm{x}$ is obtained by adding all {$\bm{x}_v, v \in [1, V]$}:
\begin{align}
\bm{x} = \sum_{v=1}^{{V}} \bm{x}_v.
\label{equ:trans_x}
\end{align}

Then, the encoded signal, $\bm{x}$ is transmitted over the SIMO-OFDM channel with CSI, $\bm{H}$, generated according to \eqref{equ:def_csi} and the received signal, $\bm{Y} \in \mathbb{C}^{N_r\times N}$, is fed to the NOS decoder which is detailed as follows. 

\textbf{NOS decoding:}
After passing the channel, the receiver first performs MMSE equalization to obtain an estimate, $\hat{\bm{x}}$ of the transmitted signal, $\bm{x}$ using CSI.  The MMSE equalization is performed for each of the OFDM subcarriers and for the $n$-th subcarrier, we have:
\begin{align}
\tilde{x}_n = (\bm{h}_n^\dagger \bm{h}_n + \sigma_d^2)^{-1} \bm{h}_n^\dagger \bm{y}_n,
\label{eq:MMSE_EQU}
\end{align}
where $\bm{y}_n, \bm{h}_n \in \mathbb{C}^{N_r}$ denote the received signal and CSI at the $n$-th OFDM subcarrier, respectively and $\sigma_d^2$ is the noise variance defined in \eqref{equ:real_yd}.
After obtaining the equalized signal, $\tilde{\bm{x}}$, we convert it into a real-valued vector, $\hat{\bm{x}} \in \mathbb{R}^{2N}$ for the subsequent NOS decoder which adopts a set of decoding functions, denoted by $g_v(\cdot), v \in [1, V]$, and each of the function is parameterized by the same neural network architecture which mirrors that of the NOS encoder.   The decoding function, $g_v$ corresponding to a specific $v$ produces a length-$D$ probability vector {$\bm{p}_v = {g_v}(\bm{\hat{x}})$} where {${p}_v[d]$} represents the probability of {${o}_v[d]=1$}.

The training of the encoding and decoding functions, ${f_v}$ and ${g_v}$, of the NOS code, is performed by optimizing the  cross-entropy loss for each pair of the one-hot input {$\bm{o}_v$} and the probability vector {$\bm{p}_v$}. The total loss is the summation of the pairwise losses:
\begin{align}
 \mathcal{L}_c = -\sum_{v=1}^{V} {\sum_{d=1}^{D}{{o}_v[d] \log({p}_v[d])}}.
\label{eq:classify_loss}
\end{align}
{We randomly generate different CRC coded bit sequences $\bm{c}$ and SIMO-OFDM channel realizations $\bm{H}$ in the training phase, and adopt the ADAM optimizer to minimize the loss in \eqref{eq:classify_loss} corresponding to different $\bm{c}$ and $\bm{H}$ realizations to train the proposed neural networks.}


\subsubsection{Achieving different trade-offs}
The NOS code discussed above is trained to minimize the decoding error as indicated in the loss function, \eqref{eq:classify_loss}. However, we note that, in the considered CIPSAC scenario, the sensing capability of the designed waveform should also be taken into account. It is shown in \cite{tit_isac} that there is a trade-off between the sensing and communication performances. In particular, deterministic signal with constant amplitude waveform serves as a good candidate for sensing, yet stochastic signal with fluctuating amplitude is more suitable for communications thanks to their capability to carry more information.
Following this intuition, we try to design different NOS codes which achieve different sensing and decoding performances.

To fulfill this, we introduce a loss function corresponding to the sensing loss. In particular, we add a constraint to the transmitted signal, $\bm{x}$, defined in \eqref{equ:trans_x} to have less fluctuated amplitude compared to the constant (unitary) amplitude signals:
\begin{equation}
    \mathcal{L}_s = \sum_{i=1}^N (|{x}[i]|^2 - 1)^2.
    \label{equ:sense_loss}
\end{equation}
To achieve different trade-offs between the sensing and the decoding performances, a hyper parameter, $\lambda$, is introduced and the overall loss function can be expressed as:
\begin{equation}
    \mathcal{L} = \mathcal{L}_c + \lambda \mathcal{L}_s.
    \label{equ:overall_loss}
\end{equation}
It is easy to understand that, the NOS framework with a large $\lambda$ value would focus on producing the codewords with near constant amplitudes leading to superior sensing performance, while the framework with a small $\lambda$ value would mainly focus on the PER performance. 
We train different NOS models with different $\lambda$ values and the training process is identical to the aforementioned NOS code.
We will show in the experiments that adopting a suitable $\lambda$ value is essential for superior sensing and decoding performances.

\begin{remark}
    We note that the design of the sensing loss in \eqref{equ:sense_loss} is largely based on the intuitions which does not guarantee optimality. A more advanced scheme involves directly adopting the MSE defined in \eqref{equ:overall_mse} as $\mathcal{L}_s$ and optimize \eqref{equ:overall_loss} using the updated $\mathcal{L}_s$. However, the acquisition of the updated $\mathcal{L}_s$ also includes non-differential operations and is impossible to optimize in an end-to-end fashion. 
\end{remark}

\subsection{CRC-Assisted NOS Code}\label{sec:crc_sparc}
We illustrate the encoding and decoding processes of the CRC-assisted NOS code with a learned codebook. Note that the processes are the same for the original SPARC code where each entry follows an i.i.d. complex Gaussian distribution.
\subsubsection{NOS encoding}
The encoding process is comprised of two parts, namely, the CRC and NOS encoding. To start with, the input bit sequence, denoted by $\bm{b} \in \{0, 1\}^{N_b}$, is fed to the CRC encoding block where $N_{crc}$ extra CRC bits are appended to $\bm{b}$ to produce a new bit sequence $\bm{c} \in \{0, 1\}^{V\log_2(D)}$ for the subsequent NOS encoding where $V\log_2(D) = N_b + N_{crc}$. Then, $\bm{c}$ is partitioned into $V$ blocks, each is comprised of $\log_2(D)$ bits.\footnote{We assume $D$ is a power of 2.} The NOS code is characterized by its codebook, $\bm{\mathcal{C}} \in \mathbb{C}^{V\times D \times N}$, which is obtained according to \cite{mimo_nos}.
The $\log_2(D)$ bits of the $v$-th block is mapped to an index,  $d_v \in [1, D]$, which corresponds to the length-$N$ sub-codeword, $\bm{\mathcal{C}}_v[d_v]$. By adding up all $V$ sub-codewords, the transmitted codeword is obtained as:
\begin{equation}
\bm{x} = \sum_{v = 1}^V \bm{\mathcal{C}}_v[d_v].
\label{equ:sparc_enc}
\end{equation}
We can easily verify that $\mathbb{E}(\bm{x}^\dagger \bm{x}) = N$ using the i.i.d. property of the elements within the codebook.

\subsubsection{NOS decoding}
Then, we present channel decoding over an SIMO-OFDM channel for the proposed CIPSAC system. We provide a brief overview of the decoding algorithm here, and refer the readers to \cite{mimo_nos} for more details.

\textbf{Codebook update:}
We assume the CSI, denoted as $\bm{H} \in \mathbb{C}^{N_r \times N}$, is perfectly available at the decoder.\footnote{The SPARC decoding algorithm is still valid when there is an estimation error in $\bm{H}$.} The codebook, $\bm{\mathcal{C}}$, is updated w.r.t. $\bm{H}$ as:
\begin{align}
\widetilde{\bm{\mathcal{C}}}_{v, \bm{H}}[d, r] &= \bm{H}[r, :] \odot \bm{\mathcal{C}}_v[d], \quad 
\notag \\
\bm{\mathcal{C}}_{v, \bm{H}}[d] &= \text{vec}(\widetilde{\bm{\mathcal{C}}}_{v, \bm{H}}[d, :])
\label{equ:update_Ch}
\end{align}
where $\widetilde{\bm{\mathcal{C}}}_{\bm{H}} \in \mathbb{C}^{V\times D \times N_r \times N}$ denotes the intermediate codebook, $\text{vec}(\cdot)$ represents the vectorization operation and the updated codebook, $\bm{\mathcal{C}}_{\bm{H}} \in \mathbb{C}^{V \times D \times N_r N}$.

\textbf{$K$-best decoding:}
We then illustrate the procedure to decode SPARC codeword. To be precise, the posterior probability of the SPARC codeword can be expressed as:
\begin{align}
    P({d_1}&,\ldots, {d_{V}}|\bm{y}, \bm{H}) \propto \notag \\ &\exp\{-\frac{1}{2\sigma_d^2} ||\bm{y}-\sum_{v=1}^{V}{\bm{\mathcal{C}}_{v,\bm{H}}[d_v]}||^2_2\},
\label{eq:joint prob}
\end{align} 
where $\bm{y} = \text{vec}(\bm{Y})$ and is of length $N_r N$.
The decoding objective is to find a combination of indices, $\{d_1, d_2, \ldots, d_V\}$ that maximizes \eqref{eq:joint prob}, or equivalently, minimizes the L2 distance, $||\bm{y}-\sum_{v=1}^{V}{\bm{\mathcal{C}}_{v,\bm{H}}[d_v]}||^2_2$.

Similar with \cite{mimo_nos}, we define the \textit{score metric}, $s^{(l)} = {||\bm{y}-\sum_{i=1}^l{\bm{\mathcal{C}}_{i, \bm{H}}[d_i]}||^2_2}$, which can be expressed recursively as:
\begin{align}
      s^{(l)} &= s^{(l-1)}  + ||\bm{\mathcal{C}}_{l, \bm{H}}[d_l]||_2^2 
  \notag \\  &+ 2\Re(\bm{\mathcal{C}}_{l, \bm{H}}^\dagger[d_l]\bm{u}^{(l-1)}- \bm{\mathcal{C}}_{l, \bm{H}}^\dagger[d_l]\bm{y}),
\label{eq:recursive}
\end{align} 
where {{$\bm{u}^{(l-1)} = \sum_{i=1}^{l-1}{\bm{\mathcal{C}}_{i,\bm{H}}[d_i]}$}} is the \textit{cumulative vector}. It is easy to show that $s^{(V)}$ is the L2 distance of interest, i.e., $s^{(V)} = ||\bm{y}-\sum_{v=1}^{V}{\bm{\mathcal{C}}_{v,\bm{H}}[d_v]}||^2_2$. 

The maximum a posterior (MAP) solution can be obtained by checking all possible $\{d_1, d_2, \ldots, d_V\}$ combinations leading to an overwhelmingly high complexity. To this end, we  introduce the $K$-best decoding algorithm as a low complexity alternative. In particular, for each layer, we only preserve $K$ candidates, and prune the others. The process starts from the root of the tree with a score initialized to $s^{(0)} = 0$. In the $l$-th layer, the $k$-th surviving node in the previous layer with accumulated indices $(d_1^k , \ldots, d_{l-1}^k)$ is extended to $D$ child nodes and the score of its $d_l$-th child, can be calculated as:
\begin{align}
         s^{(l)}_{d^k_1,\ldots,d^k_{l-1},d_l} &= s^{(l-1)}_{d^k_1,\ldots,d^k_{l-1}}  + ||\bm{\mathcal{C}}_{l,\bm{H}}[d_l]||_2^2 \notag \\
   &+ 2\Re({\bm{\mathcal{C}}}^\dagger_{l,\bm{H}}[d_l]\bm{u}^{(l-1)}_k-\bm{\mathcal{C}}^\dagger_{l,\bm{H}}[d_l]\bm{y}),
\label{eq:kbest_metric}
\end{align}
where  {{$\bm{u}^{(l-1)}_k = \sum_{i=1}^{l-1}{\bm{\mathcal{C}}_{i,\bm{H}}[d^k_i]}$}}. After obtaining $KD$ metrics belonging to all the $K$ surviving nodes in the previous layer, we select the $K$ smallest candidates and prune the others. By iteratively extending and pruning the tree, we can obtain $K$ candidates at the final layer, where the $k$-th candidate can be characterized by its accumulated indices denoted by $(d_1^k, \ldots, d_V^k)$. It is also associated with its accumulated vector, $\bm{u}_k^{(V)}$, and the metric, $s_k^{(V)} \triangleq s_{d^k_1,\ldots,d^k_{V}}^{(V)}$.

It is worth mentioning that the decoding order of the layers makes a difference and we adopt the \textit{per-layer} sorting in \cite{mimo_nos}, where the layers with more `reliable' candidates should be decoded earlier to prevent error propagation to the remaining layers. As a result, the true decoding order, $\mathcal{O}$, is a permutation of $[1, 2, \ldots, V]$.

We also implement `looped $K$-best decoding', proposed in \cite{mimo_nos}, with $N_i^d$ extra iterations to further improve the PER performance. In particular, the looped K-best decoder takes the output of the original K-best decoding algorithm as input and outputs the updated indices of the $K$ candidates, which can be expressed as:
\begin{align}
\{(d_1^k, &\ldots, d_V^k)\}_{k \in [1, K]} = \mathrm{loopedKbest}( \notag \\ 
&\{(d_1^k, \ldots, d_V^k), \bm{u}^{(V)}_k, \bm{s}_k^{(V)}\}_{k \in [1, K]}, \mathcal{O}, N_i^d).
\label{equ:looped_kbest}
\end{align}
We refer the interested readers to \cite{mimo_nos} for more details of the looped K-best decoding procedure. 

Finally, for the $k$-th candidate of the {looped K-best} decoder,  we convert each of its $V$ indices, $(d_1^k, \ldots, d_V^k)$ into a bit sequence, denoted by $\hat{\bm{c}}^v_k, v \in [1, V]$. Then, the overall decoded bit sequence corresponding the the $k$-th candidate is obtained by concatenating all $V$ bit sequences together, $\hat{\bm{c}}_k = (\hat{\bm{c}}_k^1, \ldots, \hat{\bm{c}}_k^V)$. We apply CRC to each of the $K$ bit sequences, $\hat{\bm{c}}_k, k \in [1, K]$, in sequential order and the first bit sequence which passes the CRC will be served as the final decoded output, $\bm{\hat{c}}$. If none of the $K$ bit sequences passes the CRC, we set the error flag, $T$ to 0, otherwise set $T = 1$.
The entire CRC-aided $K$-best decoding algorithm is summarized in Algorithm \ref{algorithm:decode}. Note that we use the notation $\text{idx}(k)$ to represent the indices, $(d_1^k, \ldots, d_V^k)$ in the algorithm.

\begin{remark}
    We note that the NOS decoders illustrated in Section \ref{sec:nos_code} differ from the aforementioned CRC-assisted $K$-best decoding process. This is due to the fact that the CRC-assisted $K$-best decoding involves iteratively extending and pruning the branches of the $K$-best tree which is non-differential for back propagation of the gradients. We leave it as a future work to train a better NOS codebook for superior decoding performance.
\end{remark}


\newcolumntype{Y}{>{\centering\arraybackslash}X}

\begin{table}[t]

\caption{{List of key variables.}}
\centering
 \begin{tabularx}{\linewidth}{| p{1.25cm} | Y |} 
 \hline

 Variable & Description \\  \hline

 $N, N_G$ & The number of OFDM subcarriers and guard subcarriers in one OFDM symbol.  \\  \hline
 $M$ & Minimum number of OFDM symbols to achieve the desired Doppler resolution.  \\  \hline
 $M_p, M_d$ & The number of OFDM symbols for the pilot and the data packets.  \\  \hline
 $N_r$ & Number of antennas at the receiver.\\ \hline
 {$L$} & Number of targets/paths.  \\  \hline
 {$N_i$} & Number of iterations for the iterative parameter sensing and channel decoding algorithm.  \\  \hline
 {$K$} & The number of survivors kept in each layer of the $K$-best tree search. \\  \hline
 {$N_i^d$} & Number of additional iterations for the looped $K$-best algorithm.  \\  \hline
 {$V, D$} & Number of encoder-decoder pairs and the dimension for each one-hot input, $\bm{o}_v$, adopted in the NOS scheme.  \\  \hline
 {$N_b, N_{crc}$} & The length of the information bit sequence, $\bm{b}$, and the number of CRC bits.  \\  \hline
 {$\sigma^2_p, \sigma^2_d$} &  The noise variances for the pilot and data OFDM symbols defined in \eqref{equ:def_snr} and \eqref{equ:real_yd}, respectively. \\  \hline

 $\alpha, \tau, \mu, \theta$ & The RCS, delay, Doppler and AoA parameters for a sensing target. \\  \hline 
 $\lambda$ & The hyper parameter to achieve different trade-offs between sensing and communication for NOS code. \\  \hline 
 $\bm{\mathcal{{C}}}$ & The NOS codebook which can be either randomly generated or obtained via end-to-end training. \\  \hline 
 $\bm{X}, \bm{X}_p, \bm{X}_d$ & The entire transmitted signal, $\bm{X}$, is obtained by concatenating the pilot and data signals, $\bm{X}_p$ and $\bm{X}_d$.\\  \hline 
 $\bm{Y}, \bm{Y}_p, \bm{Y}_d$ & The entire received signal, $\bm{Y}$, is obtained by concatenating the received pilot and data signals, $\bm{Y}_p$ and $\bm{Y}_d$. \\  \hline 
 \end{tabularx}
\label{tab:list_variables}

\end{table}

\section{Numerical Experiments}\label{sec:experiment}
We illustrate the effectiveness of the proposed CIPSAC system and the IPSCD algorithm under different experimental setups. Unless otherwise mentioned, we set the entries of the pilot signal $\bm{X}_p$ to unitary and the number of guard subcarriers equals to $N_G = N/4$.
An 11-bit CRC with generator polynomial, $x^{11}+x^{10}+x^9+x^5+1$ is adopted for all coded modulation schemes. Moreover,  we set the number of surviving candidates, $K = 16$ at each layer and the number of extra iterations, $N_i^d = V$, for the SPARC/NOS decoding. Finally, the  training parameters of the NOS code is identical to that in \cite{mimo_nos}.

\subsection{Static SISO}
In the static SISO setup, the scatters are assumed to be static and the receiver is equipped with a single antenna, thus, only two parameters, the RCS, $\alpha_\ell$ and delay, $\tau_\ell$ need to be estimated.  We perform extensive numerical experiment in this simplified setup to justify the gain by introducing channel codes for integrated passive sensing and communications compared with existing uncoded schemes \cite{two_stage_ipsac, isac_learn}. Ablation studies are also carried out to figure out the effect of different system parameters on the sensing and communication performances which guide the parameter selection in the more challenging mobile SISO and mobile SIMO scenarios.
We consider short packet transmission where each data packet is comprised of 13 information bits with the aforementioned 11-bit CRC. The parameters of the SPARC are set to $V = 3, D = 256, N = 32$. The number of OFDM subcarriers is $N = 32$ and $L = 3$ scatters are considered.  Without loss of generality, we assume $\alpha_\ell \in \mathcal{CN}(0, \frac{1}{L}), \ell \in [1, L]$ and $n_\ell = \ell$, i.e., the delay occupies the first $L$ taps.

\begin{figure*}
     \centering
     \begin{subfigure}{0.64\columnwidth}
         \centering
         \includegraphics[width=\columnwidth]{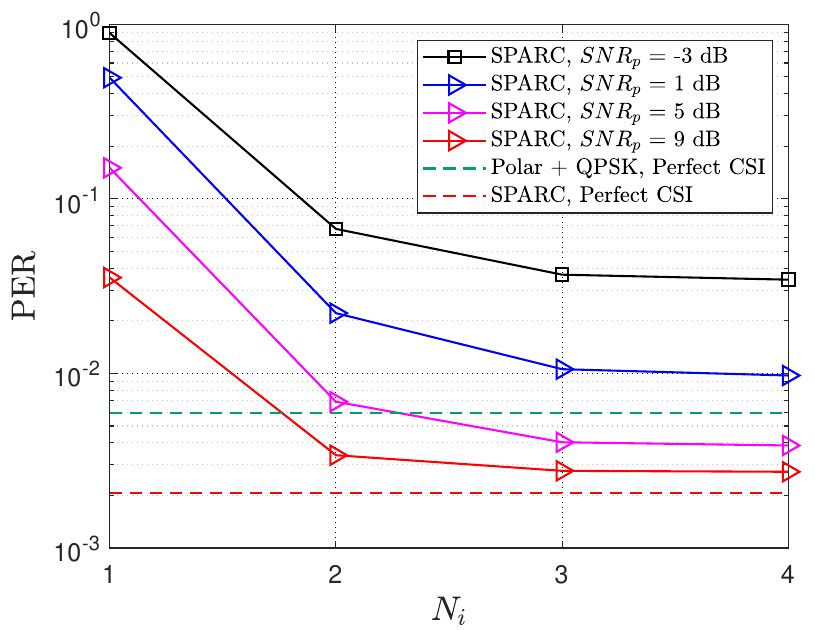}
         \caption{}
     \end{subfigure}
     \begin{subfigure}{0.62\columnwidth}
         \centering
         \includegraphics[width=\columnwidth]{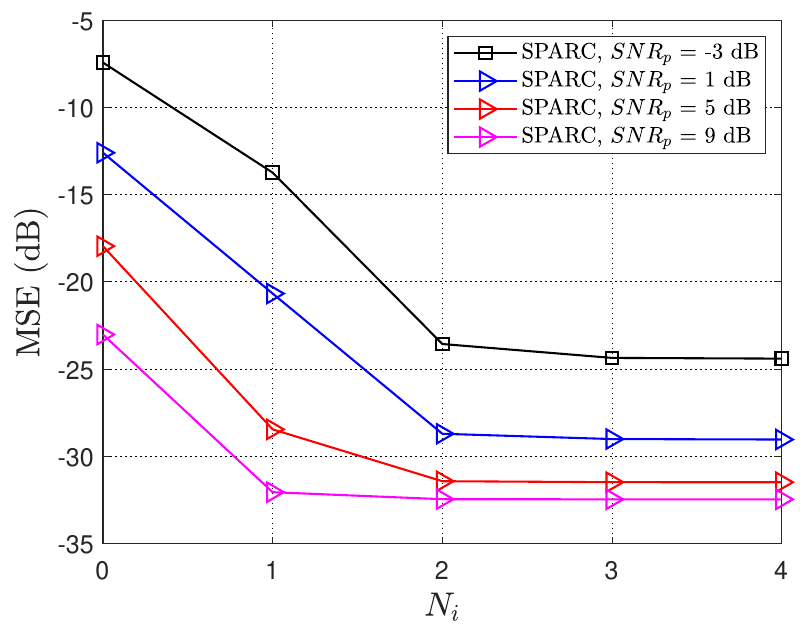}
         \caption{}
     \end{subfigure}
     \begin{subfigure}{0.66\columnwidth}
         \centering
         \includegraphics[width=\columnwidth]{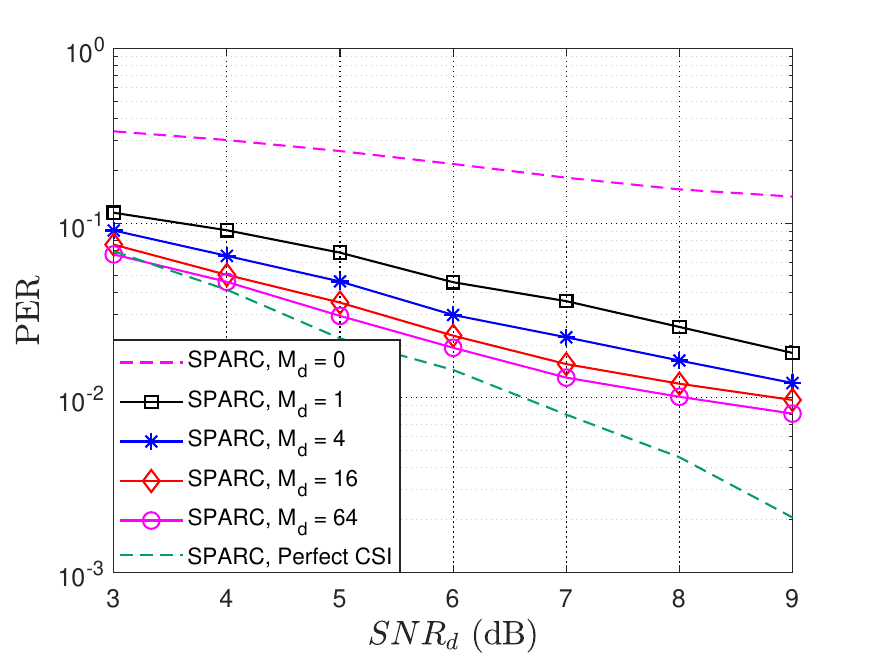}
         \caption{}
     \end{subfigure}
  \caption{Performance evaluation of the proposed CIPSAC system in the statio SISO scenario: (a) \& (b) the PER and MSE performances versus the number of iterations, $N_i$, with different $\mathrm{SNR}_p$ values and $M_d = 6$; (c) the PER performance with different numbers of blocks, $M_d$, where $\mathrm{SNR}_p = 1$ dB and $N_i = 4$.}
\label{fig:final_simu}
\end{figure*}

\begin{figure}
\centering
\includegraphics[width=0.9\columnwidth]{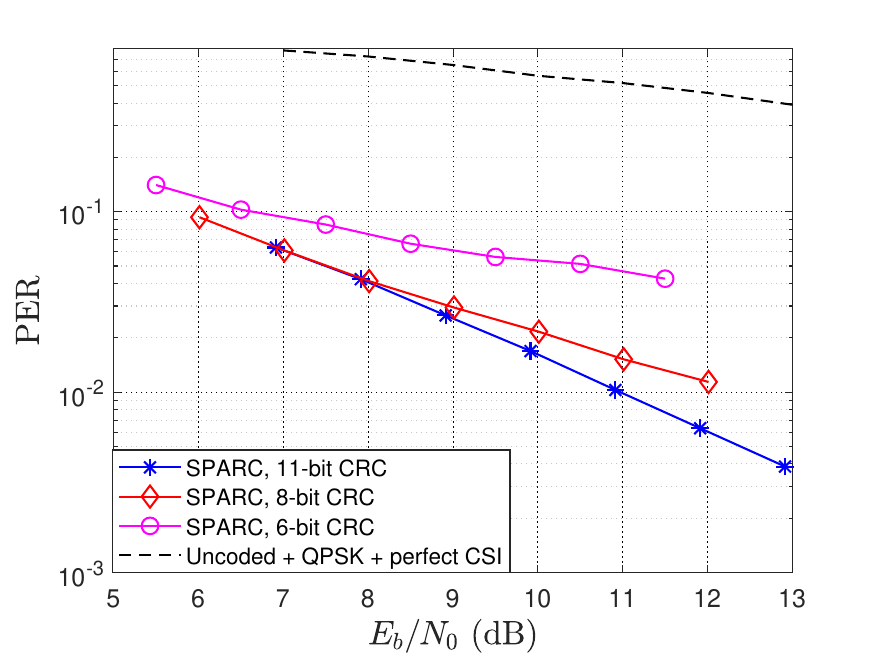}
\caption{The PER performance of the proposed CIPSAC system adopting SPARC code  is evaluated with different numbers of CRC bits. We also provide an uncoded benchmark to outline the effectiveness of the channel codes. It is worth emphasizing that we adopt $E_b/N_0$ as the x-axis for a fair comparison of the schemes with different spectrum efficiencies.}
\label{fig:abalation}
\end{figure}

\subsubsection{Performance improvement w.r.t. $N_i$}
We first illustrate the improvement brought by the number of iterations. In this simulation, different pilot SNR values, $\mathrm{SNR}_p \in \{-3, 1, 5, 9\}$ dB are considered with fixed $\mathrm{SNR}_d = 9$ dB and the number of pilot frames and data packets are set to $M_p = 1$ and $M_d = 6$, respectively. The PER and MSE performances are evaluated after collecting enough number of packet errors.   It can be seen from Fig. \ref{fig:final_simu} (a) and (b) that both the PER and the MSE improve significantly w.r.t $N_i$. In particular, Fig. \ref{fig:final_simu} (a) shows for all $\mathrm{SNR}_p$ values, the PER performance converges with merely $3$ iterations. We can also observe that the PER performance of the SPARC with $\mathrm{SNR}_p = 9$ dB and $N_i = 4$ nearly approaches the PER performance of the SPARC with perfect CSI. 

We also provide the PER performance of the Polar baseline to outline the effectiveness of the proposed SPARC. The Polar baseline utilizes the same 11-bit CRC and adopts SCL decoding algorithm \cite{scl_polar}. The decoded bit sequence is obtained after the MMSE equalization and Polar decoding. As can be seen in Fig. \ref{fig:final_simu} (a), the proposed SPARC outperforms the Polar baseline in terms of PER.

Fig. \ref{fig:final_simu} (b) manifests the sensing MSE w.r.t. $N_i$. Note that in this simplified static SISO scenario, the estimated delay $\hat{\tau}$ seldom deviates from the ground truth, ${\tau}$. Thus, we adopt the MSE defined in Eqn. (7) in \cite{crc_isac} which also takes the RCS parameter, $\alpha$, into account. The MSE is large when $N_i = 0$ which is due to the fact that we only use one pilot frame to estimate the parameters. When $N_i$ grows, a larger number of data packets are successfully decoded which are utilized for improved MSE performance.

\subsubsection{The effect of different number of data packets} 
{{we investigate the effect of different number of data packets, $M_d$, on the system performance. In this simulation, the pilot SNR is fixed at $\mathrm{SNR}_p = 1$ dB, while $M_d$ is selected from $M_d \in \{1, 4, 16, 64\}$. Intuitively, increasing $M_d$ enhances the probability of correctly decoding a larger number of data packets, which can then serve as additional pilot symbols, thereby improving the accuracy of parameter estimation. However, a larger $M_d$ also introduces higher latency, as the latency scales proportionally with $M_d$.}}

As can be seen in Fig. \ref{fig:final_simu} (c), the PER performance improves rapidly from $M_d = 0$ to $M_d = 4$ yet it nearly ceases to improve when $M_d \ge 16$. It is also shown in the figure that there is a gap between $M_d = 64$ and the SPARC with perfect CSI which is due to the CRC outage: the wrong SPARC re-encoded codewords would degrade the parameter estimate which hinders the final PER performance. This will be verified in the subsequent simulations concerning the outage probability, $P_o$.

\begin{table}[tbp]
\caption{The outage probabilities, $P_o$, for three CRC settings with different SNR values.}
\begin{center}
\begin{tabular}{c|c|c|c}
\hline

\cline{1-4} 
$\mathrm{SNR}_d$ (dB) & \textbf{3}& \textbf{6}& \textbf{9} \\
\hline

$P_o$ w/ 6-bit CRC & 0.102  & 0.058 & 0.040 \\
$P_o$ w/ 8-bit CRC & 0.028  & 0.015 & 0.009 \\
$P_o$ w/ 11-bit CRC & $2.2\times 10^{-3}$ & $1.7 \times 10^{-3}$  & $1.2\times 10^{-3}$ \\

\hline
\end{tabular}
\label{tab:crc_res}
\end{center}

\end{table}

\begin{figure*}
     \centering
     \begin{subfigure}{0.64\columnwidth}
         \centering
         \includegraphics[width=\columnwidth]{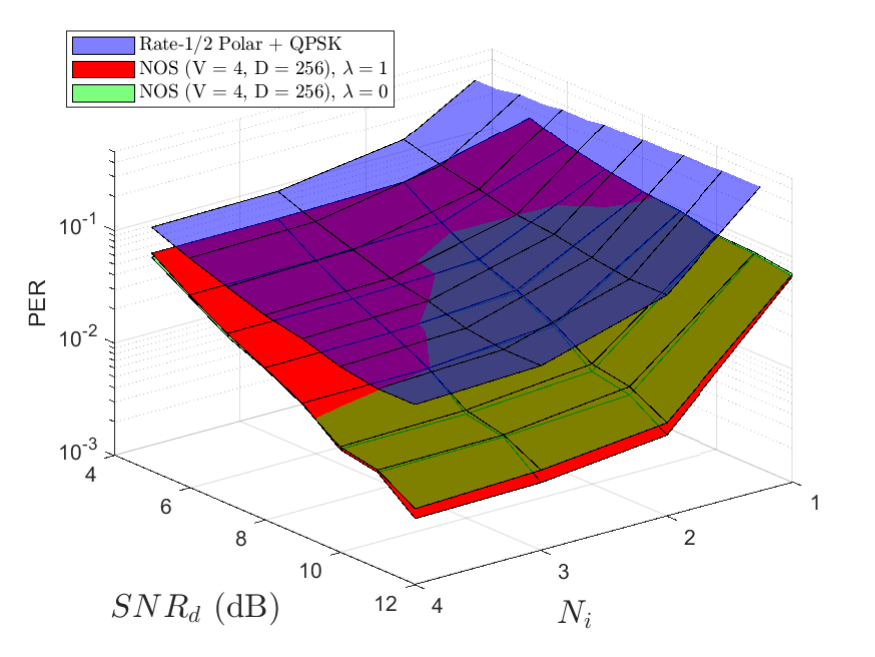}
         \caption{}
     \end{subfigure}
     \begin{subfigure}{0.64\columnwidth}
         \centering
         \includegraphics[width=\columnwidth]{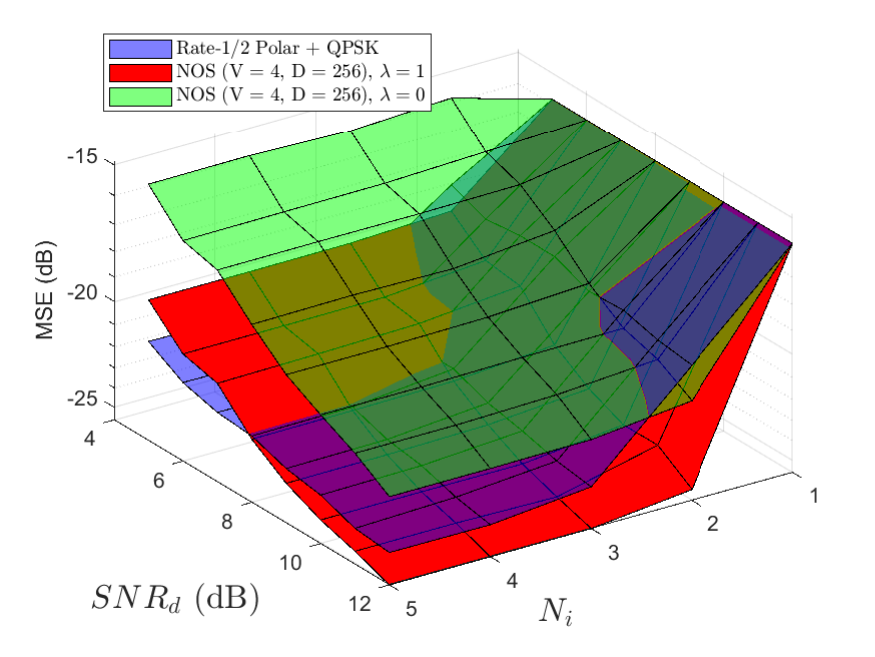}
         \caption{}
     \end{subfigure}
     \begin{subfigure}{0.64\columnwidth}
         \centering
         \includegraphics[width=\columnwidth]{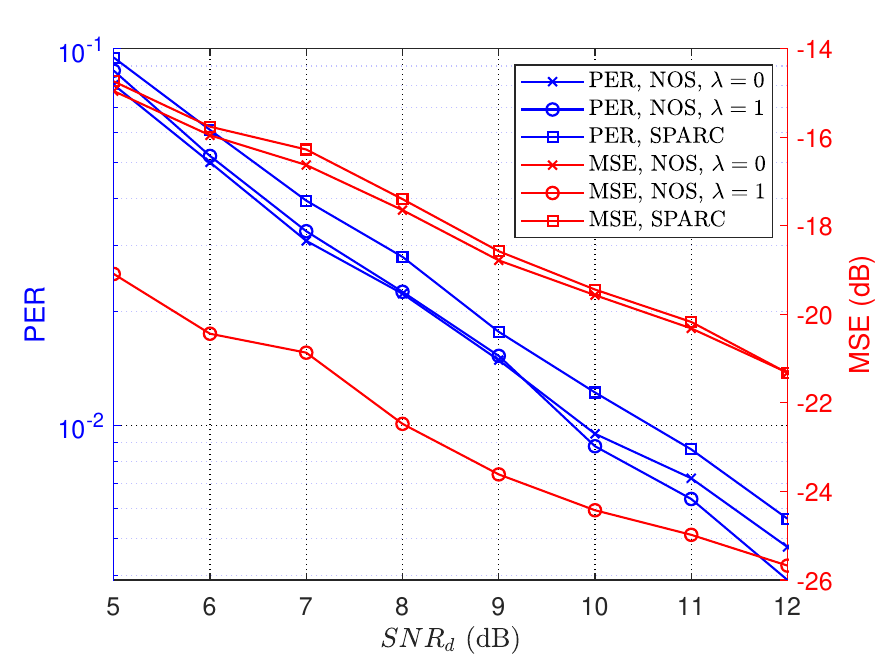}
         \caption{}
     \end{subfigure}
  \caption{Performance evaluation of the proposed IPSCD algorithm in the mobile SISO scenario. (a) \& (b) Comparison between the NOS codes trained under different $\lambda$ values and the Polar baseline in terms of PER and MSE. (c) The superiority of the NOS codes compared to the SPARC codes which adopt random Gaussian codebook \cite{SPARC2014}.}
\label{fig:simu_siso_2d}
\end{figure*}

\begin{figure}
\centering
\includegraphics[width=0.9\columnwidth]{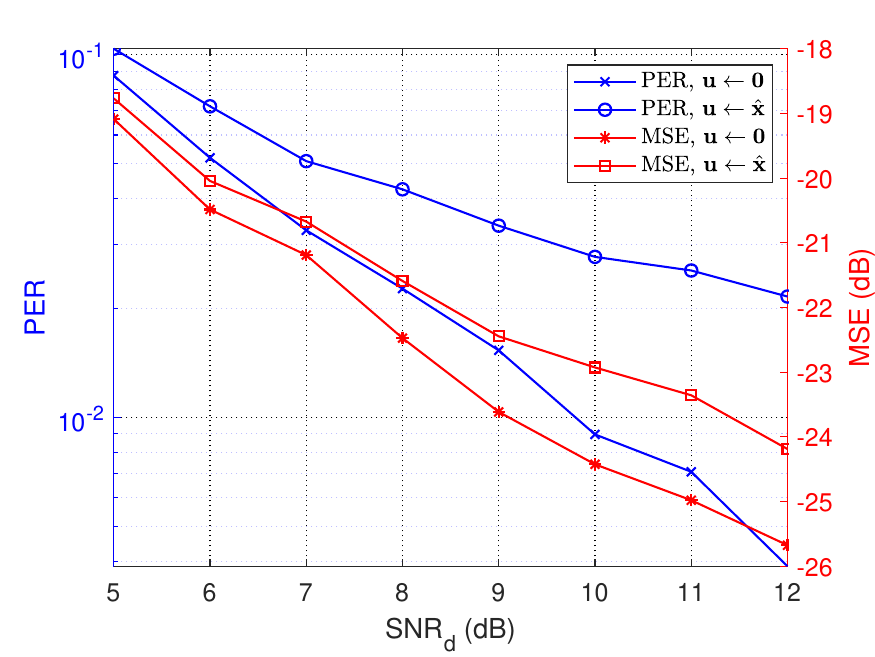}
\caption{The PER and MSE performances obtained by two different $\bm{u}$ realizations.}
\label{fig:u_policy}
\end{figure}

\begin{figure*}[t]
\centering
\includegraphics[width=\linewidth]{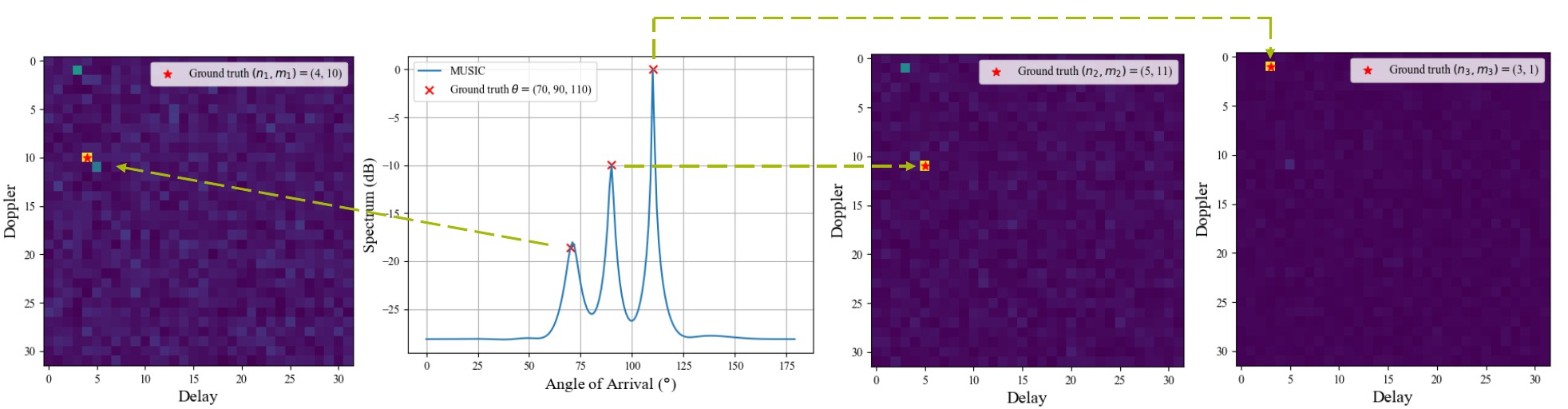}\\
\caption{Numerical illustration of the sensing algorithm shown in Fig. \ref{fig:sensing_flowchart}. The ground truth parameters are set to $\bm{\theta} = [70, 90, 110], \bm{n} = [4, 5, 3], \bm{m} = [10, 11, 1]$. It is observed from the left figure that accurate $n_1, m_1$ estimation are obtained even if $\hat{\theta}_1 = 71$ slightly deviates from $\theta_1 = 70$.}
\label{fig:sensing_visualize}
\end{figure*}

\begin{figure*}
     \centering
     \begin{subfigure}{0.64\columnwidth}
         \centering
         \includegraphics[width=\columnwidth]{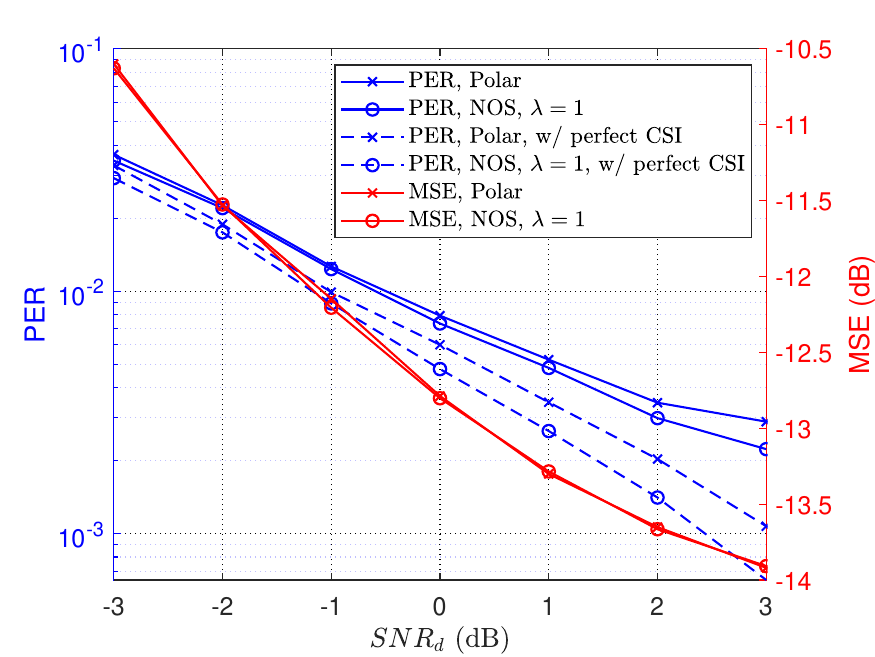}
         \caption{}
     \end{subfigure}
     \begin{subfigure}{0.64\columnwidth}
         \centering
         \includegraphics[width=\columnwidth]{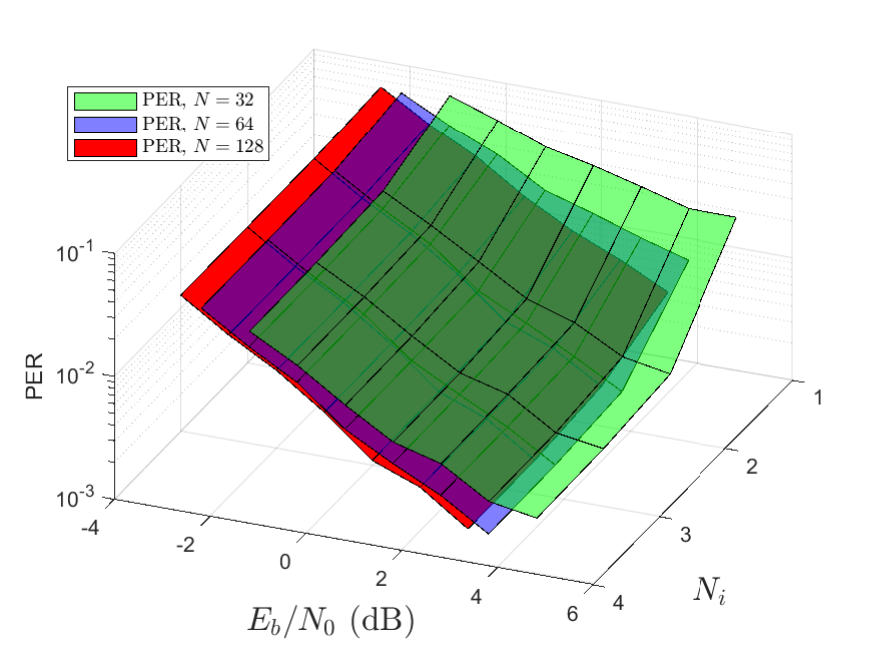}
         \caption{}
     \end{subfigure}
     \begin{subfigure}{0.64\columnwidth}
         \centering
         \includegraphics[width=\columnwidth]{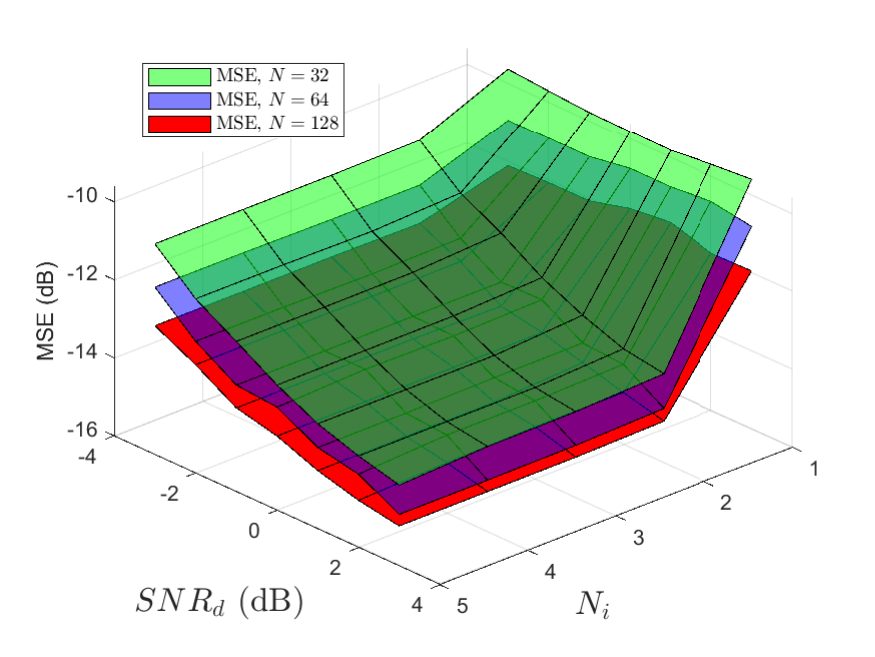}
         \caption{}
     \end{subfigure}
  \caption{Performance evaluation of the proposed IPSCD algorithm in the mobile SIMO scenario. (a) We show the comparison between the NOS and the Polar baseline in terms of PER and MSE for $N = 32$. (b) \& (c) For larger $N$ values, we adopt the Polar code with QPSK modulation and report its PER and MSE results.}
\label{fig:simu_simo_3d}
\end{figure*}

\subsubsection{Ablation studies} 
Finally, we provide ablation studies for a comprehensive understanding of the proposed IPSCD system under static SISO scenario. In particular, we fix the SPARC parameters, i.e., $(V = 3, D = 256, N = 32)$ and change different CRC settings. In particular, we consider 6-bit and 8-bit CRC with generator polynomials, $x^6 + x^5 + 1$ and $x^8 + x^2 + x + 1$, respectively. The PER performance of the uncoded QPSK symbols with perfect CSI is also provided to show the necessity of introducing the channel codes. For all the schemes, we set the number of blocks, $M_d = 6$ and the pilot SNR, $\mathrm{SNR}_p = 5$ dB. It is worth mentioning that, to provide a fair comparison for the schemes with different CRC lengths (leading to different spectrum efficiencies), we use $E_b/N_0$ instead of $\mathrm{SNR}$ as the x-axis.

As shown in Fig. \ref{fig:abalation}, without the aid of channel codes, the PER of the uncoded QPSK symbols is much higher than that using channel codes.
It is also shown that having a shorter CRC would harm the system performance, this is due to the fact that the outage probabilities, $P_o$, for both the 6-bit and 8-bit CRC settings are relatively high as shown in Table \ref{tab:crc_res}. This would lead to wrong SPARC re-encoded symbols, $\bm{\hat{x}}_b$, and harms the subsequent parameter estimation and CSI reconstruction performance. The 11-bit CRC has a much smaller $P_o$, thus, it has the best PER performance.

Finally, we provide the outrage probability, $P_o$, w.r.t. the channel SNR of the data OFDM symbols, $\mathrm{SNR}_d$ which is shown in Table \ref{tab:crc_res}. It can be seen that for all the three different CRC setups,  $P_o$  decreases when  $\mathrm{SNR}_d$ increases which is because the proposed SPARC decoding algorithm produces more reliable bit sequences with higher $\mathrm{SNR}_d$, thus, the bit sequence which passes the CRC would be more likely to be the ground truth.

\subsection{Mobile SISO}
In this scenario, an extra Doppler parameter, $\mu_\ell$, needs to be estimated which is assumed to be randomly distributed, i.e., $m_\ell \in [-\frac{M}{2}, \frac{M}{2}-1]$ where $M = 32$.  The parameters for the SPARC/NOS code are set to $V = 4, D = 256, N = 32$. 

\subsubsection{Performance comparison with different channel codes}
We first evaluate the MSE and PER performances achieved by different channel codes. In this simulation, the Polar code and the NOS code trained at different $\lambda \in \{0, 1\}$ values are considered. The numbers of OFDM symbols for the pilot and data packets are set to $M_p = 4$ and $M_d = 28$, respectively. Moreover, we set the SNR of the pilot signal to $\mathrm{SNR}_p = 7$ dB and vary the SNR of the data signal $\mathrm{SNR}_d \in [5, 12]$ dB. 
Their relative PER and MSE performances are shown in Fig. \ref{fig:simu_siso_2d} (a) \& (b), respectively. As can be seen, for all the considered channel coded modulation schemes, both the communication and sensing performances improve with more iterations $N_i$ and higher $\mathrm{SNR}_d$ values.\footnote{{This is only except for the NOS code trained under $\lambda = 0$ and evaluated at low $\mathrm{SNR}_d = 5$ dB where the MSE increases w.r.t. the sensing performance obtained from the pilot signal (i.e., $N_i = 0$). This can be explained by the fact that the codeword belonging to the NOS code trained under $\lambda = 0$ deviates a lot from the constant amplitude codeword leading to sub-optimal sensing performance.}}  This verifies the effectiveness of the proposed CIPSAC system in the mobile SISO scenario.  
It is also observed that the NOS code with $\lambda = 1$ outperforms the Polar baseline in terms of PER for all $\mathrm{SNR}_d$ and $N_i$ values. As for the MSE metric, it is shown that thanks to the constant amplitude waveform, the Polar baseline slightly outperforms the NOS code with $\lambda = 1$ when $\mathrm{SNR}_d$ is small, yet it is surpassed when $\mathrm{SNR}_d$ becomes larger. This is due to the fact that the NOS code exhibits much lower PER with a high $\mathrm{SNR}_d$ value leading to a larger number of correctly decoded data symbols for a more accurate parameter estimation.  The NOS code with $\lambda = 0$, on the other hand, achieves comparable PER performance w.r.t. that with $\lambda = 1$ yet the worst sensing performance. In particular, the NOS code with $\lambda = 0$ yields superior PER performance with low $\mathrm{SNR}_d$ value compared to that with $\lambda = 1$ because it is optimized to achieve the best decoding performance. However, due to its poor sensing capability, it is eventually outperformed by the $\lambda = 1$ model under high $\mathrm{SNR}_d$ value. As a result, the NOS code with $\lambda = 1$ provides a significantly improved sensing performance without harming the PER performance. Thus, we adopt the NOS code trained under $\lambda = 1$ throughout the paper.

We then show the performance improvement brought by the NOS code over the SPARC code which adopts complex Gaussian codebook. We adopt the same simulation setup in Fig. \ref{fig:simu_siso_2d} (a) \& (b) and set the number of iterations, $N_i = 4$. Both the PER and the MSE performances are shown in Fig. \ref{fig:simu_siso_2d} (c) which correspond to the left and the right part of the figure, respectively. It can be observed that the NOS codes with different $\lambda \in \{0, 1\}$ values outperform the SPARC code in terms of both sensing and decoding performances demonstrating the effectiveness of the end-to-end optimization.

\subsubsection{Evaluation of different $\bm{u}$ implementations}
We then show the effectiveness of the proposed $\bm{u}$ implementation illustrated in Proposition \ref{pro:pro1}. In this simulation, we compare the sensing and decoding performances achieved by the proposed $\bm{u} \leftarrow \bm{0}_N$ to a naive protocol where $\bm{u}$ is replaced by the re-encoded codeword, $\hat{\bm{x}}$, from a wrongly decoded bit sequence, $\hat{\bm{c}}$ with $T = 0$. The simulation setup is identical to that in Fig. \ref{fig:simu_siso_2d}. As can be seen in Fig. \ref{fig:u_policy}, both the PER and the MSE performances of the proposed protocol outperform that of the naive one. It is especially the case in terms of PER where a 4 dB SNR gain can be observed when PER equals to $3\times 10^{-2}$. Thus, we adopt the protocol introduced in Proposition \ref{pro:pro1} throughout the paper.

\subsection{Mobile SIMO}
Finally, we show the effectiveness of the proposed IPSCD algorithm in the mobile SIMO scenario where the generation of the RCS, delay and Doppler parameters are identical to that in the mobile SISO scenario, while the AoA values, $\theta_\ell$, are sampled uniformly within the range of $[\frac{\pi}{6}, \frac{5\pi}{6}]$.

\subsubsection{Illustrations of the sensing procedure}
To start with, we first demonstrate the sensing procedure, $h(\cdot)$, introduced in Section \ref{sec:sense_alg} via numerical experiment. In this simulation, we adopt $M_p = 32$ OFDM symbols each contains $N = 32$ OFDM subcarriers as pilot signal.  Note that only the pilot signal is used for sensing and the number of data symbols is set to $M_d = 0$. Moreover, we set the number of receive antennas to $N_r = 8$, the delay and Doppler resolution to $M = N= 32$ and the pilot SNR to $SNR_p = 0$ dB.

As can be seen in Fig. \ref{fig:sensing_visualize}, the ground truth AoA, delay and Doppler parameters are $\bm{\theta} = \{70, 90, 110\}, \bm{n} = \{4, 5, 3\}, \bm{m} = \{10, 11, 1\}$, respectively.\footnote{Here we assume each entry of $\bm{m}$ ranges from $[1, M]$ for better illustration, which differs from the definition in \eqref{equ:real_dd}.} As illustrated in Section \ref{sec:sense_alg}, we first produce the estimated AoA values, $\bm{\hat{\theta}} = \{71, 90, 110\}$ via MUSIC algorithm as shown in the middle left of the figure. Then, for each of the estimated ${\hat{\theta}}_\ell$, we apply the data substitution, equalization and the 2D-FFT algorithm to produce the corresponding delay and Doppler parameters which are compared with the ground truth shown in the remaining plots. 
It's worth mentioning that, for all the estimated AoA values, the corresponding $\hat{\tau}_\ell$ and $\hat{\mu}_\ell$ match with the ground truth even there is an estimation error where $\hat{\theta}_1 \neq \theta_1$.  This shows the effectiveness of the sensing algorithm in Section \ref{sec:sense_alg}.

\subsubsection{Performance comparison of different channel codes}
We then consider the performance comparison of the Polar code and the NOS code trained under different $\lambda$ values in the mobile SIMO setup. In this simulation, we adopt the same setup as in Fig. \ref{fig:sensing_visualize}  except the number of OFDM symbols for the pilot and data transmission are set to $M_p = 4, M_d = 28$, respectively. The pilot SNR is set to $\mathrm{SNR}_p = 0$ dB and the data SNR, $\mathrm{SNR}_d \in [-3, 3]$ dB. The parameters of the Polar and the NOS code are identical to that in the mobile SISO scenario.

The Polar code and the NOS code trained under $\lambda = 1$ is compared in Fig. \ref{fig:simu_simo_3d} (a) where both channel codes adopt $N_i = 4$. As can be seen, both schemes achieve nearly the same sensing performance, yet the NOS code outperforms the Polar baseline in terms of PER by a small margin. The PER performance of the two channel codes with perfect CSI are shown in the dashed curves. As can be seen, the proposed NOS code outperforms the Polar baseline with an approximately $0.5$ dB SNR gain when $\text{PER} \approx 10^{-3}$.

\subsubsection{Performance evaluation for different block lengths} We then illustrate the sensing and decoding performances of the CIPSAC system with IPSCD algorithm for the channel codes with different block lengths. As illustrated in \cite{nos, mimo_nos}, the SNR gain of the NOS code shrinks w.r.t. the Polar coded modulation under larger block length. Thus, in this simulation, we adopt the rate-1/2 Polar code with QPSK modulation and set the block lengths, $N = 32, 64, 128$ which correspond to the numbers of information bits, $N_b = 21, 53, 117$, respectively. Their relative PER and MSE performances are shown in Fig. \ref{fig:simu_simo_3d} (b) \& (c), respectively. Note that we use $E_b/N_0$ as the x-axis to evaluate PER performance for a fair comparison.

As can be seen, both the PER and MSE performances improve with larger $N$. It aligns with the intuition for the PER performance, as the error correction ability of the channel codes improves w.r.t. the block length. For the sensing performance, the SIMO-OFDM signal with larger block length ($N$) provides more ($N(M_p+M_d)$) data symbols for accurate AoA estimation which further improves the quality of the subsequent delay, Doppler and RCS estimation leading to lower sensing MSE. Finally, we note that though increasing the block length improves both the sensing and communication performances, it also causes higher time delay and computational complexity when performing the channel decoding algorithm.

\section{Conclusion}
In this paper, we propose a novel CIPSAC system with OFDM, where the BS passively senses the channel parameters and decodes the information of the user simultaneously using both the pilot and data symbols from the user. We consider the SIMO-OFDM scenario where the channel is consisted of multiple mobile targets and the BS is equipped with multiple antennas. 
Upon receiving the pilot and data signals, a novel IPSCD algorithm is proposed, where the correctly decoded codewords which pass the CRC are utilized to improve the sensing performance at the BS. Detailed illustrations of the sensing procedure which outputs satisfactory parameter estimates are provided which serves as the core component for the IPSCD algorithm.
To further improve the channel decoding performance, we adopt the learning-based NOS code and introduce a new loss function to train NOS encoder for different sensing and communication trade-offs. 
Simulation results show the effectiveness of the proposed CIPSAC system and the IPSCD algorithm over different experimental setups, where both the sensing and communication performances are significantly improved with a few iterations.  We also carry out ablation studies for a comprehensive understanding of the proposed scheme.

\bibliographystyle{IEEEbib}
\bibliography{refs}

\appendices
\section{On the optimality of the $\bm{u}$ implementation in Proposition \ref{pro:pro1}}\label{sec:APPA}
This appendix provides detailed derivations on the optimality of the proposed data substitution procedure shown in Proposition \ref{pro:pro1} of Section \ref{sec:sense_alg}. We consider a scenario where the total number of transmitted OFDM symbols is identical to the Doppler resolution, i.e., $M_p + M_d = M$, the receiver has a single antenna and the transceiver communicates over the wireless channel induced by a single target with delay $n_0 \in [0, N_G-1]$ and Doppler $v_0 \in [-\frac{M}{2}, \frac{M}{2}-1]$. Note that the proof is also applicable to the multiple target scenario by applying successive interference cancellation (SIC) algorithm. As a result, the channel described in \eqref{equ:simo_channel} degrades to:
\begin{equation}
    Y_{n, m} = a X_{n, m} e^{j 2\pi \frac{n_0 n}{N}} e^{-j 2\pi \frac{v_0 m}{M}} + W_{n, m},
    \label{equ:uni_tgt_channel}
\end{equation}
where $a \sim \mathcal{CN}(0, 1), W_{n, m} \sim \mathcal{CN}(0, \sigma^2)$ denote the complex gain and the i.i.d. AWGN term, respectively.
It is worth mentioning that the transmitted signal $\bm{X} = [\bm{X}_p, \bm{X}_d]$ is comprised of the pilot and data OFDM symbols where the pilot signal, $\bm{X}_p$, has unitary entry, i.e., $\bm{X}_p[n,m] = 1$ and the data signal, $\bm{X}_d[n,m]$ has zero mean and unit variance. We note that this property can be theoretically/numerically verified for the SPARC code adopting random Guassian codebook/learning-based NOS code.
Without loss of generality, we assume $(n_0, v_0) = (0, 0)$.

The peak to side lobe ratio (PSR) defined as
\begin{equation}
    \text{PSR} \triangleq \mathbb{E}\left(|r_{0, 0}|^2 - \sum_{(n_g, v)\neq (0, 0)} |r_{n_g, v}|^2 \right),
    \label{equ:PSR}
\end{equation}
is adopted to evaluate different $\bm{u}$ implementations in \eqref{equ:def_tilde_Z} where $r_{n_g, v}$ is the matched filter output for the $(n_g, v)$-th grid on the delay-Doppler plane given by:
\begin{equation}
    r_{n_g, v} \triangleq \sum_{m,n} Y_{n, m} S^*_{n, m}e^{j 2\pi \frac{n_g n}{N}} e^{-j 2\pi \frac{vm}{M}},
    \label{equ:r_nv}
\end{equation}
where $n_g \in [0, N_G-1], v \in [-\frac{M}{2}, \frac{M}{2}-1]$ while $\bm{S} \in \mathbb{C}^{N\times M}$ is the reference signal for the matched filter.

As illustrated in Section \ref{sec:CIPSAC}, due to the decoding error, the receiver might not have full access to $X_{n, m}$. In particular, we consider the scenario where a single packet error occurs in the data OFDM symbols with index $m^* \in [M_p, M-1]$ and the corresponding CRC flag, $T_{m^*} = 0$. After introducing the substituted data packet, $\bm{u} \in \mathbb{C}^{N}$, which is randomly generated at the receiver, the  reference signal, $\bm{S} \in \mathbb{C}^{N \times M}$, is defined as:
\begin{equation}
    S_{n, m} = \left\{
    \begin{aligned}
    X_{n, m}  & , & m \neq m^*, \\
    u_n & , & m = m^*.
    \end{aligned}
    \right.
    \label{equ:def_Skm}
\end{equation}
Substituting \eqref{equ:uni_tgt_channel} and \eqref{equ:def_Skm} into \eqref{equ:r_nv}, and note that $(n_0, v_0) = (0, 0)$, we have:

\begin{align}
    &|r_{n_g, v}|^2 = 
    \left|\sum_{m,n} a X_{n, m} S^*_{n, m} e^{j 2\pi \frac{n_g n}{N}} e^{-j 2\pi \frac{vm}{M}}\right|^2 + \notag \\
    &2\Re\left({\sum_{\substack{n, n' \\ m, m'}} |a|^2 X_{n, m} S^*_{n, m} W^*_{n', m'} S_{n', m'} t_{n_g}(n, n') t^*_v(m, m')}\right) \notag \\
    &+\sum_{\substack{n, n' \\ m, m'}} W_{n, m} S^*_{n, m} W^*_{n', m'} S_{n', m'} t_{n_g}(n, n') t^*_v(m, m'),\label{equ:r_nv_square}
\end{align}
where $t_{n_g}(n, n') \triangleq e^{j 2\pi \frac{n_g(n-n')}{N}}, t_v^*(m, m') \triangleq e^{-j 2\pi \frac{v(m-m')}{M}}$.

Taking expectation over \eqref{equ:r_nv_square}, it can be immediately obtained that the second term in \eqref{equ:r_nv_square} is zero due to $\mathbb{E}(W_{n, m}) = 0$ and $\bm{W}$ is independent with $\bm{X}$ and $\bm{S}$. The expectation of the remaining terms can be expressed as:
\begin{align}
    \mathbb{E} |r_{n_g, v}|^2 & =  \mathbb{E}\left|\sum_{m,n} a X_{n,m} S^*_{n,m} e^{j 2\pi \frac{n_g n}{N}} e^{-j 2\pi \frac{vm}{M}}\right|^2 \notag \\
    & + \sigma^2 (M-1)N + \sigma^2 \sum_{n} \mathbb{E}|u_{ n}|^2,
    \label{equ:E_r_nv}
\end{align}
which is obtained by noticing $\mathbb{E}(W_{n, m}W^*_{n', m'}) = \sigma^2 \delta(n-n', m-m')$ and
\begin{equation}
    \mathbb{E}(S_{n, m}S^*_{n', m'}) = 
    \begin{cases}
    0, & m \neq m', \\
    1, & n=n', m=m'\neq m^*, \\
    \mathbb{E}(|u_n|^2)  , & n=n', m=m'= m^*.
    \end{cases}
    \label{equ:sum_smk}
\end{equation}

The first term in \eqref{equ:r_nv_square} can be further expressed as:
\begin{align}
    &\left|\sum_{n, m\neq m^*} a p_{n, m} e^{j 2\pi \frac{n_g n}{N}} e^{-j 2\pi \frac{vm}{M}} \right. \notag \\
    &+ \left. \sum_n a X_{n, m^*}u_n e^{j 2\pi \frac{n_g n}{N}} e^{-j 2\pi \frac{vm^*}{M}}\right|^2 \notag \\
    =  &\sum_{\substack{n, n' \\ m, m'\neq m^*} }   \underbrace{|a|^2 p_{n, m} p_{n', m'} t_{n_g}(n,n') t_v^*(m, m')}_{\triangleq 
 \; O_{n_g,v}} \notag\\
    &+ 2\Re \left(|a|^2\sum_{\substack{n, n' \\ m\neq m^*}} p_{n, m} X^*_{n', m^*}u^*_{n'} t_{n_g}(n,n') t_v^*(m, m^*) \right) \notag \\
    &  + |a|^2 \left|\sum_{n} X_{n, m^*}u_{n} e^{j 2\pi \frac{n_g n}{N}} \right|^2,
    \label{equ:term1}
\end{align}
where we introduce $p_{n, m} \triangleq |X_{n, m}|^2$ for simplicity.
It can be noticed that the first term is independent w.r.t. different $\bm{u}$ implementations. For the second and the third terms, their expectations are zero and $|a|^2 \sum_n \mathbb{E}|u_n|^2$ using the fact that $X^*_{n', m^*}$ is independent with $p_{n, m}, u^*_{n'}$ and $\mathbb{E}(X^*_{n', m^*}) = 0$.

By substituting \eqref{equ:term1} and \eqref{equ:E_r_nv} into \eqref{equ:PSR}, we obtain:
\begin{align}
    \text{PSR} &= (O_{0,0} - \sum_{(n_g, v) \neq (0, 0)} O_{n_g, v}) - (MN_G-1)(M-1)N\sigma^2 \notag \\
    & -(N_G M-1)(\sigma^2 + |a|^2) \sum_n \mathbb{E}(|u_n|^2),
\end{align}
where the first and the second terms are independent with $\bm{u}$. It is obvious that to maximize the PSR value, we should set $\bm{u} = \bm{0}_N$ as illustrated in Proposition \ref{pro:pro1} as the third term is  always non-negative. This finishes the proof of Proposition \ref{pro:pro1}.
\end{document}